\documentclass[pre,aps,twocolumn,amsmath,amssymb,showpacs,superscriptaddress]{revtex4}
\usepackage[T1]{fontenc}
\usepackage[latin9]{inputenc}
\usepackage{color}
\usepackage{amsmath}
\usepackage{amssymb}
\usepackage{graphicx}

\makeatletter
\@ifundefined{textcolor}{}
{%
 \definecolor{BLACK}{gray}{0}
 \definecolor{WHITE}{gray}{1}
 \definecolor{RED}{rgb}{1,0,0}
 \definecolor{GREEN}{rgb}{0,1,0}
 \definecolor{BLUE}{rgb}{0,0,1}
 \definecolor{CYAN}{cmyk}{1,0,0,0}
 \definecolor{MAGENTA}{cmyk}{0,1,0,0}
 \definecolor{YELLOW}{cmyk}{0,0,1,0}
}

\pacs{03.65.Aa, 03.65.Ta, 03.67.Mn}
\newcommand{\scalar}[2]{\langle #1 | #2 \rangle}
\newcommand{\ketbra}[2]{| #1 \rangle \langle #2 |}
\newcommand{\ket}[1]{| #1 \rangle}
\newcommand{\bra}[1]{\langle #1 |}

\newcommand{\1}{{\rm 1\hspace{-0.9mm}l}}

\newcommand{\id}{\1}

\newcommand{\Cplx}{\mathbb{C}}
\newcommand{\Real}{\mathbb{R}}

\newtheorem{prop}{Proposition}
\newtheorem{theorem}{Theorem}
\newtheorem{lemma}{Lemma}
\newtheorem{definition}{Definition}
\newtheorem{conjecture}{Conjecture}
\newtheorem{corollary}{Corollary}

\begin{document}

\date{ver. 20, July 14, 2015}

\title{Certainty relations, mutual entanglement and non-displacable
manifolds}

\author{Zbigniew Pucha\l a}

\affiliation{Institute of Theoretical and Applied Informatics, Polish Academy
of Sciences, Ba\l tycka 5, 44-100 Gliwice, Poland}

\affiliation{Institute of Physics, Jagiellonian University, ul Reymonta 4, 30-059
Krak\'ow, Poland}

\author{\L ukasz Rudnicki}

\affiliation{Institute for Physics, University of Freiburg, Rheinstra\ss e 10, D-79104
Freiburg, Germany}

\affiliation{Center for Theoretical Physics, Polish Academy of Sciences, Aleja
Lotnik\'ow 32/46, PL-02-668 Warsaw, Poland}

\author{Krzysztof Chabuda}

\affiliation{Center for Theoretical Physics, Polish Academy of Sciences, Aleja
Lotnik\'ow 32/46, PL-02-668 Warsaw, Poland}

\author{Miko\l aj Paraniak}

\affiliation{Center for Theoretical Physics, Polish Academy of Sciences, Aleja
Lotnik\'ow 32/46, PL-02-668 Warsaw, Poland}

\author{Karol \.{Z}yczkowski}

\email{karol.zyczkowski@uj.edu.pl}

\affiliation{Institute of Physics, Jagiellonian University, ul Reymonta 4, 30-059
Krak\'ow, Poland}

\affiliation{Center for Theoretical Physics, Polish Academy of Sciences, Aleja
Lotnik\'ow 32/46, PL-02-668 Warsaw, Poland}
\begin{abstract}

We derive explicit bounds for the average entropy
characterizing measurements of a pure quantum state of size $N$
in $L$ orthogonal bases. Lower bounds lead to novel entropic uncertainty relations,
while upper bounds allow us to formulate universal certainty relations.
For $L=2$ the maximal average entropy saturates at $\log N$
as there exists a mutually coherent state, 
but certainty relations 
are shown to be nontrivial for  $L \ge 3$ measurements.
In the case of a prime power dimension, $N=p^k$, 
and the number of measurements $L=N+1$, 
the upper bound for the average entropy becomes minimal
for a collection of mutually unbiased bases. 
Analogous approach is used to study entanglement with respect
to $L$ different splittings of a composite system, linked by bi-partite
quantum gates. We show that for any two-qubit unitary gate $U\in \mathcal{U}(4)$ there exist states  being mutually separable or mutually entangled with respect to
both splittings (related by $U$) of the composite system. The latter statement follows from
the fact that the real projective space $\mathbb{R}P^{3}\subset\mathbb{C}P^{3}$
is non-displacable. For $L=3$ splittings the maximal sum of $L$
entanglement entropies is conjectured to achieve its minimum for a collection of three
mutually entangled bases, formed by two mutually entangling gates. 
\end{abstract}
\maketitle
\section{Introduction}

Quantum uncertainty relations characterizing the ultimate limitations
\cite{Lahti,IBBLR} which nature puts on the preparation and measurements
of any quantum state continuously attract a significant attention.
Much of it is focused on the entropic formulation of uncertainty
\cite{deVicente,deVicenteComm,PRZ13,FGG13,CP14,RPZ14,Korzekwa1,Bosyk1,Bosyk2,Bosyk3,Kaniewski,Banach},
since the information entropy function is a clever collection of the
information contained in all the moments of the probability
distribution.
Standard entropic uncertainty relations \cite{BBM75,Deutsch,MU88}
provide a lower bound for the sum of entropies characterizing information
obtained in two arbitrary orthogonal measurements. Various generalizations
including positive-operator valued measures (POVM) \cite{Bosyk3},
coarse graining \cite{IBB1,IBB2,OptCon,LR2015}, quantum memory
\cite{Berta,ColFur},
different trade-off relations \cite{Roga,Rastegin}, or even quasi-hermitian
operators \cite{quasi} and elaborate studies devoted to quantum protocols
\cite{Protokoly}, can be found in the literature related to the topic
discussed in this paper.

On the contrary, 
no comparable 
effort had been made in
order to establish relevant \emph{certainty relations} \cite{Sa93,Sa95}
given as an upper bound for the sum of the two entropies in question.
Even though, for more than two measurements described in terms of
mutually unbiased bases (MUB) almost optimal entropic certainty relations
have been derived \cite{Sa93,Sa95,Molmer,IEEE2001,Wehner}, there were
no corresponding results valid for two arbitrary measurements. The
two very recent contributions \cite{KJR14,Wolf} independently solved
that long-standing 
problem. 

An equator of a sphere is non-displacable as any two
great circles do intersect.
Such a statement can be generalized in various ways for higher dimensions.
Making use of the fact that a great torus $T_{N-1}$ embedded
in a complex projective space $\mathbb{C}P^{N-1}$
is nondisplacable with respect to transformations 
by a unitary $U\in \mathcal{U}(N)$ \cite{Cho04}, 
it is possible to show that there exists a quantum state 
{\sl mutually unbiased} with respect to both bases. 
 We shall further refer to such kind of state as  being {\sl mutually coherent}.

In this work we analyze upper bounds for the average entropy involving an arbitrary number of orthogonal measurements. We 
derive a universal certainty relation valid for any set of $L$ measurements in an $N$ dimensional Hilbert space.
Assuming that $N$ is a power of prime,
we further analyze the case of mutually unbiased bases, for which
we conjecture that the difference between the upper
and the lower limits is the smallest among all orthogonal 
measurements in $N+1$ bases.
An analogous statement that the variance of the Shannon entropy 
is minimal for MUBs is based on numerical results, while
a counterpart proposition for the Tsallis entropy of order two (also called the linear entropy)
is analytically proven.

The parallel aim of the paper is to study 
average entanglement of a given bipartite quantum state,
computed with respect to an arbitrary collection of $L$ 
splittings of the Hilbert space, related by global unitary matrices.
Relaying on the fact that the real projective space 
$\mathbb{R}P^{3}$ is non-displacable in $\mathbb{C}P^{3}$
we show that for any two splittings of ${\cal H}_4$
into two subspaces of size two, there exists a 
{\sl mutually entangled state}, maximally entangled
with respect to both partitions.
Numerical results allow us to conjecture 
that the same statement can be true for $N \times N$ systems. 

This work is organized as follows.
In section II we introduce necessary notation
and discuss trivial certainty relations for two orthogonal measurements. 
Some consequences of this result are further investigated in Section III.
In Section \ref{sec:mutual-coherence}  we discuss  mutual
coherence in the situation with more than two measurements and derive the certainty and uncertainty relations relevant for  any choice of $L$ and $N$. 
In Sections \ref{sec:mutually-entangled-states} and \ref{sec:Mutually-entangling-gates}
we explore connections between the concept of mutually coherent states
and quantum entanglement by searching for  mutually entangled states and 
mutually entangling gates.

\section{Certainty relations and mutually coherent states}

Consider a quantum state $|\psi\rangle\in{\cal H}_{N}$
belonging to an $N$-dimensional Hilbert space $\mathcal{H}_{N}$
which is measured in several orthonormal bases determined by unitary
matrices $\left\{ U_{k}\right\} $. These bases (each of them forms
the columns of a particular $U_{k}$), are eigenbases of some observables
standing behind the measurements. Information gained in that process
can be described by the Ingarden--Urbanik entropy \cite{IU62} 
\begin{equation}
S^{IU}(|\psi\rangle,U_{k})=S_{k}=-\sum_{i=1}^{N}p_{i}^{(k)}\log p_{i}^{(k)},
\end{equation}
which is the Shannon entropy calculated for the probability distribution
$p_{i}^{(k)}=|\langle i|U_{k}|\psi\rangle|^{2}$. 
The choice of the base of the logarithm is arbitrary, but in numerical 
calculations we will use natural logarithms.
The entropy $S_{k}$
is a non-negative quantity upper-bounded by $\log N$.
The question
about the uncertainty and certainty relations for the two measurements
(given in terms of $U_{1}$ and $U_{2}$) is devoted to the two numbers
(or rather functions of $N$ and $U_{2}U_{1}^{\dagger}$) $B_{{\rm min}}$
and $B_{{\rm max}}$ to be determined, such that
\begin{equation}
0\le B_{{\rm min}}\le \frac{S_{1}+S_{2}}{2}\le B_{{\rm max}}\le \log N.
\end{equation}
The matrix $U_{2}U_{1}^{\dagger}$ is a single quantity that matters here
since by the transformation $\ket{\psi}\mapsto U_1^\dagger \ket{\psi'}$ one can always bring the first unitary to be the identity $\1$,
while all other matrices become  multiplied by $ U_{1}^{\dagger}$.
A profound (though very rarely close to optimal) example of a valid
$B_{{\rm min}}$ is $-\max_{i,j}\log c_{ij}$, with $c_{ij}$ denoting
the modulus of the matrix element of $U_{2}U_{1}^{\dagger}$, situated
in $i$-th row and $j$-th column. This is the well known Maassen-Uffink
result \cite{MU88}.

\begin{figure}
\centering
\includegraphics[width=\linewidth]{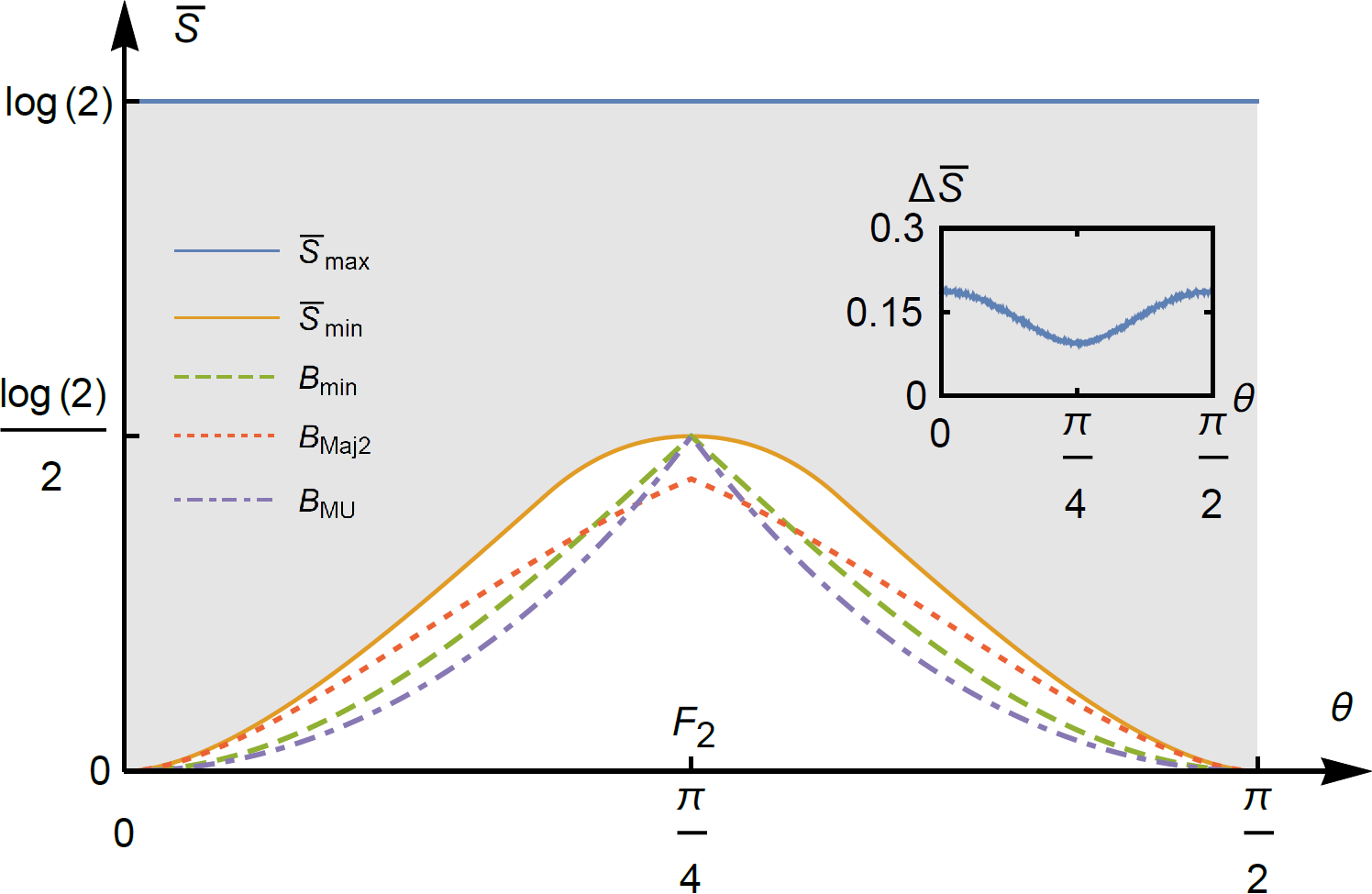}
\caption{Average entropy ${\bar S}$
for two single--qubit measurements related by an orthogonal matrix
$O\in O(2)$ as a function of the rotation angle $\theta$.
Solid lines denote numerical lower and upper bounds which limit the allowed region 
(shaded), while dotted lines represent Massen Uffink and majorization lower bounds. Dashed line corresponds to the bound $B_\textrm{min}$ derived in this paper.
The inset shows root mean square of the entropy $\Delta {\bar S}$
averaged over the set of pure states as a function of the angle $\theta$.}
\label{fig:SN2L2}
\end{figure}

In Fig. \ref{fig:SN2L2} we show behavior of the
minimal and maximal average entropy  ${\bar S}=(S_1+S_2)/2$
for an orthogonal matrix $O=[\cos\theta, \sin\theta; -\sin\theta, \cos\theta]$, 
as a function of the rotation angle $\theta$. Note that the case of any unitary matrix of order $N=2$
is equivalent to a certain orthogonal matrix \cite{PRZ13}.
Even for this simple family, the minimal values 
${\bar S}_{\rm min}$ are rather cumbersome \cite{SR88}.
${\bar S}_{\rm min}$ lays obviously above  the Maassen--Uffink bound  \cite{MU88}, as well as
the majorization bound $B_\textrm{Maj2}$ derived in \cite{RPZ14}. In Fig. \ref{fig:SN2L2} we also plot in advance our candidate for $B_\textrm{min}$, which in Section IV is derived for an arbitrary setting described in general by $L$ unitaries $U_1,\ldots,U_L\in \mathcal{U}(N)$.
The lower bound assumes the largest value for $\theta =\pi/4$,
for which the matrix $O$ coincided with the Hadamard matrix.

On the other hand, the upper bound occurs to be trivial,
as the maximal value is always saturated, 
$B_{\rm max}={\bar S}_{\rm max}= \log 2$. This is a direct consequence of a more general statement mentioned already in the Introduction, saying
that for any choice of a unitary matrix
$U\in \mathcal{U}(N)$  one can always find a state $|\psi_{\rm coh}\rangle$ 
of the form $(1,e^{i\phi_{2}},\dots,e^{i\phi_{N}})/\sqrt{N}$
such that all probabilities are equal, 
$|\langle i|\psi_{\rm coh}\rangle|^2=|\langle i|U|\psi_{\rm coh}\rangle|^2=1/N$
\cite{KJR14,Wolf}.
This leads to the upper bound $B_{{\rm max}}= \log N$. 
Hence for any two orthogonal measurements in any dimension $N$
there exist no non-trivial upper bounds  and certainty relations. 
This counter-intuitive statement 
follows from the fact that the great tori $T^{N-1}$ embedded in the set of
pure states $\mathbb{C}P^{N-1}$ is nondisplacable with respect to action of
$\mathcal{U}(N)$~\cite{Cho04,Tam08}. More intuitively, it is a generalization of an easy
fact that any two great circles on a sphere do intersect. Therefore the torus $T^{N-1}$
of basis coherent states \cite{BCP14}
($\phi_{1}\equiv0$) 
\[
|\psi_{{\rm coh}}^{\boldsymbol{\phi}}\rangle=\frac{1}{\sqrt{N}}\sum_{j=1}^{N}e^{i\phi_{j}}|j\rangle,
\]
and the torus $U|\psi_{{\rm coh}}^{\boldsymbol{\phi}}\rangle$, defined
in terms of any unitary matrix $U$ (which in the above scheme is
equal to $U_{2}U_{1}^{\dagger}$), of states coherent with respect
to the transformed basis do intersect \cite{KJR14}. An arbitrary
point from the intersection represents a \textsl{mutually coherent state},
coherent with respect to both bases, also called 'zero-noise zero-disturbance
state' \cite{KJR14}.
Recent investigations show \cite{Ingem15}
that for $N=3$ both two-tori generically cross in $6$ or $4$ discrete points
-- see Fig. \ref{fig:torus2}

\begin{figure}
\begin{centering}
\includegraphics[scale=0.33]{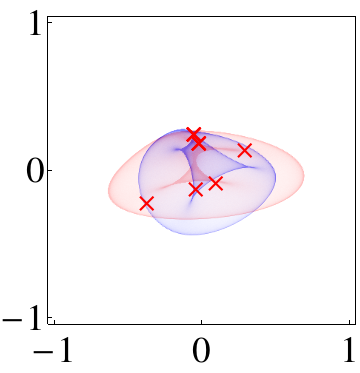}\,\,\,\includegraphics[scale=0.33]{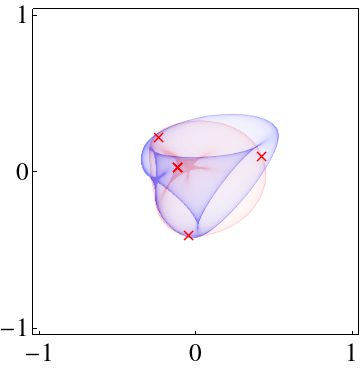}
\par\end{centering}
\protect
\caption{
Two examples of projections of two tori $T^2$ and $U (T^2)$ embedded in $\Cplx P^2$
on a plane. Here $U\in \mathcal{U}(3)$ and crosses denote the intersection points. 
}
\label{fig:torus2}
\end{figure}

 Note that the density matrix $|\psi_{{\rm coh}}^{\boldsymbol{\phi}}\rangle\langle\psi_{{\rm coh}}^{\boldsymbol{\phi}}|$
written in both bases is\emph{ }\textsl{\emph{contradiagonal}} \cite{LPZ14},
as it has all diagonal elements equal. Hence taking into account permutations of the
spectrum, the basis coherent states
are as distant from the diagonal density matrices as possible 
at a single orbit of unitarily similar states.
It is important to emphasize
here the difference between the set of basis coherent states, which
forms a torus, and the set of \textsl{\emph{spin coherent states}}
\cite{ZFG90}, or more generally $\mathcal{SU}(K)$ coherent states, producing
a complex projective space $\mathbb{C}P^{K-1}\subset\mathbb{C}P^{N-1}$.

\section{Further consequences of trivial certainty relation}

The fact that $B_{{\rm max}} = \log N$ implies few interesting consequences.
Assume that Alice possesses a maximally coherent (with respect to
her computational basis $|1\rangle,\ldots,|N\rangle$) quantum state
$|\psi_{{\rm coh}}^{\boldsymbol{\phi}}\rangle$ with tunable parameters
$\phi_{2},\ldots,\phi_{N}$. The two immediate corollaries follow:
\begin{corollary}[Sharing] For any choice of the basis on the Bob
side, Alice can always tune the phases in a way that she 
can share the maximal coherence of the state
 $|\psi_{{\rm coh}}^{\boldsymbol{\phi}}\rangle$
with Bob.
\end{corollary} 
\begin{corollary}[Recovering] In the presence
of the unitary evolution given by an arbitrary unitary operator $U\left(t\right)$,
so that the considered quantum state at any time moment $T$ is equal
to $U(T)|\psi_{{\rm coh}}^{\boldsymbol{\phi}}\rangle$, Alice can
always tune the phases in a way that she can recover the maximally
coherent state, i.e. $U(T)|\psi_{{\rm coh}}^{\boldsymbol{\phi}}\rangle$
is maximally coherent.
\end{corollary}

A physically relevant question related to the potential usefulness
of the above corollaries concerns imperfections, i.e. the case when
the angles $\phi_{2},\ldots,\phi_{N}$ are not ideally tuned. What
is then the coherence of the second (in Bob's basis or after the evolution)
quantum state? In order to answer that question we need to quantify
the coherence. To this end we resort to the $l_{1}$-norm of coherence 

\begin{equation}
C_{l_{1}}(\rho)=\sum_{i\neq j}|\rho_{ij}|=\left(\sum_{i}|\scalar{i}{\psi}|\right)^{2}-1,\label{l1Norm}
\end{equation}
which is a proper measure \cite{BCP14} of the discussed resource.
The second equality in (\ref{l1Norm}) is valid only in the case when
$\rho$ is pure. For the maximally coherent state the measure attains
its maximum equal to $N-1$. 

We start with the following simple lemma. 
\begin{lemma}
\label{lemma:coherence-bound}
Let $|\psi_{{\rm coh}}^{\boldsymbol{\phi}}\rangle$ be a maximally
coherent state 
and let $|\xi\rangle$ be
any pure state such that 
\begin{equation}
|\scalar{\psi_{{\rm coh}}^{\boldsymbol{\phi}}}{\xi}|^{2}=1-\varepsilon,\label{Error1}
\end{equation}
with some error $\varepsilon$. Then 
\begin{equation}
C_{l_{1}}(\left|\xi\right\rangle )\geq N-1-\varepsilon N.\label{bound1}
\end{equation}
\end{lemma} To prove that lemma we only notice that 
\begin{equation}
\sqrt{1-\varepsilon}=|\scalar{\psi_{{\rm coh}}^{\boldsymbol{\phi}}}{\xi}|=\frac{1}{\sqrt{N}}|\xi_{1}+\sum_{j=2}^{N}e^{-i\phi_{j}}\xi_{j}|\leq\frac{1}{\sqrt{N}}\sum_{j}|\xi_{j}|,\label{Derivation1}
\end{equation}
where $\xi_{j}=\scalar{j}{\xi}$, and rearrange the resulting
inequality using (\ref{l1Norm}).

We know that the angles in question can be tuned in a way, that for
any unitary $U$, the state 
\begin{equation}
|\tilde{\psi}_{{\rm coh}}^{\boldsymbol{\phi}}\rangle=U|\psi_{{\rm coh}}^{\boldsymbol{\phi}}\rangle,
\end{equation}
is maximally coherent. Moreover, if we prepare the state $|\psi_{{\rm coh}}^{\boldsymbol{\phi}}\rangle$
imperfectly, so that instead of it we have in our disposal a state
$|\xi\rangle$ satisfying Eq. (\ref{Error1}), then also
\begin{equation}
|\langle\tilde{\psi}_{{\rm coh}}^{\boldsymbol{\phi}}|U|\xi\rangle|^{2}=1-\varepsilon.
\end{equation}
The bound (\ref{bound1}), which is linear in the error $\varepsilon$,
thus immediately applies to the transformed state $U|\xi\rangle$.
We observe, that a general preparation-imperfection described by $\varepsilon$
linearly decreases the coherence of the quantum state. In the next
part, we show that whenever the imperfections are provided only by
the phase mismatch ($\phi_{1}=0$):
\begin{equation}
|\xi\rangle=|\psi_{{\rm coh}}^{\boldsymbol{\phi}+\boldsymbol{\chi}}\rangle,
\end{equation}
with $\left|\chi_{j}\right|\leq\varepsilon$, then the coherence decreases
by a term quadratic in $\varepsilon$. We have the chain of relations

\begin{equation}
|\scalar{\psi_{{\rm coh}}^{\boldsymbol{\phi}}}{\xi}|^{2}=\frac{1}{N^{2}}|\sum_{j}e^{i\chi_{j}}|^{2}\geq\frac{1}{N^{2}}|\sum_{j}e^{\pm i\varepsilon}|^{2},
\end{equation}
which directly lead to: 
\begin{equation}
\begin{split}|\scalar{\psi_{{\rm coh}}^{\boldsymbol{\phi}}}{\xi}|^{2} & \geq1-\sin^{2}(\varepsilon)\text{ \ \ \ \ \ \ \ for even }N\\
|\scalar{\psi_{{\rm coh}}^{\boldsymbol{\phi}}}{\xi}|^{2} & \geq1-c_{N}\sin^{2}(\varepsilon)\text{ \ \ \ for odd }N
\end{split}
,
\end{equation}
with $c_{N}=1-1/N^{2}$. These bounds, together with Lemma~\ref{lemma:coherence-bound}
show that the difference between $N-1$ and $C_{l_{1}}(U\ket{\xi})$
is proportional to $\varepsilon^{2}$.

\section{mutual coherence for several measurements\label{sec:mutual-coherence}}

\subsection{Uncertainty and certainty relations}

In a more general setup, one studies $L$ orthogonal measurements determined
by a collection of $L$ unitary matrices $\{U_{1}\equiv\1,U_{2},\dots,U_{L}\}$.
The mutual coherence together with related concepts while described
in terms of the information entropies is captured by the uncertainty and certainty
relations
\begin{equation}
0\le B_{{\rm min}}\le \frac1L \sum_{j=1}^{L}S_{j}\le B_{{\rm max}}\le \log N,\label{sumx}
\end{equation}
where as before $S_{j}$ is the Shannon entropy of the probability
distribution $p_{i}^{(j)}=|\langle i|U_{j}|\psi\rangle|^{2}$.
In general, it is not an easy task to provide non-trivial bounds $B_{{\rm min}}$,
and $B_{{\rm max}}$ valid for a broad class of measurements.
Several lower bounds, leading to uncertainty relations, were recently studied in the literature
\cite{PRZ13,FGG13,RPZ14,LMF15}, 
but we present below an alternative lower bound.
Furthermore, we derive in Section IV B a universal upper bound
which leads to a {\sl certainty} relation.
Note that the information acquired in the set of measurements can also be characterized by the average entropy
of R{\'e}nyi or Tsallis \cite{Rastegin13D}, 
which reduce to the Shannon entropy in particular cases.

For a given set of bases defining orthogonal measurements,
a natural question appears, whether the average entropy can
achieve the maximal value $\log N$. This is the case if there exists 
a mutually coherent state, 
i.e. the state $|\psi_{{\rm coh}}^{\boldsymbol{\phi}}\rangle$ such that
\begin{equation}
\frac{1}{L} 
\sum_{j=1}^{L}C_{l_{1}}(U_{j}\ket{\psi_{{\rm coh}}^{\boldsymbol{\phi}}})=
N-1
\end{equation}
is maximal. In order to answer the above question, we use the decomposition
of unitary matrices \cite{vos,Wolf}, being a corollary of the fact that
$|\psi_{{\rm coh}}^{\boldsymbol{\phi}}\rangle$ exists for $L=2$,
\begin{equation}
U_{j}=D(\boldsymbol{\omega}^{(j)})\mathcal{F}_{N}\left(1\oplus Y_{j}\right)\mathcal{F}_{N}^{\dagger}D(-\boldsymbol{\phi}^{(j)}).\label{decomposition}
\end{equation}
This parameterization involves the phase gate
\begin{equation}
D(\boldsymbol{\alpha}^{(j)})=\textrm{diag}\left(e^{i\alpha_{1}^{(j)}},e^{i\alpha_{2}^{(j)}},\ldots,e^{i\alpha_{N}^{(j)}}\right),
\end{equation}
and the Fourier matrix $\left(\mathcal{F}_{N}\right)_{kl}=e^{2i\pi\left(k-1\right)\left(l-1\right)/N}/\sqrt{N}$.
The matrices $Y_{j}$ represent arbitrary $N-1$-dimensional unitary
operations acting on the subspace spanned by $|2\rangle,\ldots,|N\rangle$.

The phase gate $D(-\boldsymbol{\phi}^{(j)})$ acting on the state
$|\psi_{{\rm coh}}^{\boldsymbol{\phi}^{(j)}}\rangle$ produces the
dephased maximally coherent state of the form $|\psi_{{\rm coh}}^{\boldsymbol{0}}\rangle=\sum_{j=1}^{N}|j\rangle/\sqrt{N}$.
Further application of the Fourier gate $\mathcal{F}_{N}^{\dagger}$
transforms the state $|\psi_{{\rm coh}}^{\boldsymbol{0}}\rangle$
into $|1\rangle$. The latter state remains unchanged if one applies
$1\oplus Y_{j}$ with an arbitrary $Y_{j}$. In the final steps, the
inverse Fourier transform together with the second phase gate in (\ref{decomposition})
leads to the final, maximally coherent state $|\psi_{{\rm coh}}^{\boldsymbol{\omega}^{(j)}}\rangle$.
We are now in position to answer the major question of this section.
\begin{corollary} 
The unitary matrices $\1,U_{2},\dots,U_{L}$ given
by the decomposition (\ref{decomposition}), such that at least one
$Y_{j}\neq\1$, possess a mutually coherent state if the phase gates
$D(-\boldsymbol{\phi}^{(j)})$ are the same for all matrices in question,
i.e. do not depend on the index $j=2,\ldots,L$. 
\end{corollary}
The above corollary is an immediate consequence of the involved decomposition.
It does not exclude other possibilities with special, $\boldsymbol{\phi}^{(j)}$-dependent, internal unitaries $Y_{j}$ allowing for different right
phase gates, but the situation described by Corollary 3 seems to be generic.

We note in passing that the concept of mutual coherence is related to unextensibility of mutually unbiased bases \cite{Grassl}. The set of MUBs is called extensible, if there exists an additional basis formed by the states being mutually coherent with respect to the MUBs in question \cite{BCWeh}. Thus, if the analysed  bases possess no mutually coherent state they are unextensible.

\subsection{Generally valid bounds}

Our major aim is to derive the bounds $B_{\textrm{min}}$
and $B_{\textrm{max}}$ relevant for a general setting (arbitrary
$L$ and $N$) described in terms of a collection of unitaries $U_{1},\ldots,U_{L}$.
To achieve that goal, we need to briefly introduce the Bloch representation
of a quantum state. Denote by $\sigma_{i}$, $i=1,\ldots,N^{2}-1$
the traceless and Hermitian generators of the group $\mathcal{SU}(N)$
fulfilling $\textrm{Tr}\sigma_{i}\sigma_{i'}=2\delta_{i'i}$, which are given by Pauli matrices for $N=2$. Any
quantum state can be spanned by a basis formed by the identity $\1_{N}$
and the matrices $\{\sigma_{i}\}$. In particular, the density matrix
of the state $\left|\psi\right\rangle $ can be written as:
\begin{equation}
\left|\psi\right\rangle \left\langle \psi\right|=\frac{1}{N}\left(\1_{N}+\sqrt{\frac{N\left(N-1\right)}{2}}\sum_{i=1}^{N^{2}-1}x_{i}\sigma_{i}\right),
\end{equation}
The Bloch vector $\boldsymbol{x}$
is constrained by $\boldsymbol{x}\cdot\boldsymbol{x}=1$ and \cite{Bloch1,Bloch2}
\begin{equation}
2\left(N-2\right)\boldsymbol{x}=\sqrt{N\left(N-1\right)/2}\,\textrm{Tr}\left((\boldsymbol{x}\cdot\boldsymbol{\sigma})\boldsymbol{\sigma}\right).\label{constraint}
\end{equation}

Let us now rescale the original probabilities $p_{i}^{(k)}$ to
be 
\begin{equation}
\tilde{p}_{i,k}=L^{-1}p_{i}^{(k)}\equiv\frac{1}{L}|\langle i|U_{k}|\psi\rangle|^{2},
\end{equation}
so that $\tilde{p}_{i,k}$ sum up (with respect to both $1\leq i\leq N$
and $1\leq k\leq L$) to $1$. In other words, we treat $L$ orthogonal measurements as a single POVM involving $N\cdot L$ Kraus operators.

Define the 'purity' coefficient

\begin{equation}
\mathcal{P}=\sum_{k=1}^{L}\sum_{i=1}^{N}\tilde{p}_{i,k}^{2},\qquad\frac{1}{LN}\leq\mathcal{P}\leq\frac{1}{L}.
\end{equation}
We shall now prove a statement crucial in the derivation of the general
bounds: \begin{theorem}\label{THP} The coefficient $\mathcal{P}$
is bounded 
\begin{equation}
\mathcal{P}_{\min}\leq\mathcal{P}\leq\mathcal{P}_{\max},
\end{equation}
by
\begin{equation}
\mathcal{P}_{\min/\max}=\frac{1}{LN}+\left(\frac{N-1}{2NL^2}\right)\mathcal{M}_{\min/\max},
\end{equation}
where $\mathcal{M}_{\min}$ and $\mathcal{M}_{\max}$ respectively
denote the minimal and the maximal eigenvalues of the matrix 
\begin{equation}
M_{j'j}=\sum_{k=1}^{L}\sum_{i=1}^{N}\textrm{Tr}\left(U_{k}^{\dagger}|i\rangle\langle i|U_{k}\sigma_{j'}\right)\textrm{Tr}\left(U_{k}^{\dagger}|i\rangle\langle i|U_{k}\sigma_{j}\right).
\end{equation}
\end{theorem} We start the proof with the chains of inequalities
defining $\mathcal{P}_{\min}$ and $\mathcal{P}_{\max}$:
\begin{equation}
\mathcal{P}\leq\max_{\left|\psi\right\rangle }\mathcal{P}\leq\max_{\boldsymbol{x}\cdot\boldsymbol{x}=1}\mathcal{P}=\mathcal{P}_{\max},
\end{equation}
\begin{equation}
\mathcal{P}\geq\min_{\left|\psi\right\rangle }\mathcal{P}\geq\min_{\boldsymbol{x}\cdot\boldsymbol{x}=1}\mathcal{P}=\mathcal{P}_{\min}.
\end{equation}
In other words, optimization with respect to $\left|\psi\right\rangle $
is equivalent to optimization made for the Bloch vector $\boldsymbol{x}$
constrained by $\boldsymbol{x}\cdot\boldsymbol{x}=1$ and (\ref{constraint}).
We obtain the desired bounds by skipping the second constraint. Due
to the fact that all matrices $\sigma_{i}$ are traceless, we have the identity
\begin{equation}
\sum_{i=1}^{N}\textrm{Tr}\left(U_{k}^{\dagger}|i\rangle\langle i|U_{k}\sigma_{j}\right)=0,
\end{equation}
leading to the dependence of $\mathcal{P}$ on $\boldsymbol{x}$, of
the form:
\begin{equation}
\mathcal{P}(\boldsymbol{x})=\frac{1}{LN}+\frac{N-1}{2NL^2}\sum_{j',j=1}^{N^{2}-1}M_{j'j}x_{j}x_{j'}.
\end{equation}
The last step of the proof is the direct optimization with respect
to $\boldsymbol{x}$, which since the matrix $M$ is hermitian, picks
up its relevant eigenvalues. 

We are now in position to present the major result of this Section,
which leads to \emph{purity-optimized} entropic uncertainty and certainty relations.

\begin{theorem} The valid bounds $B_{\min}$ and $B_{\max}$ are
of the form: 
\begin{equation}
\label{bmin}
B_{\min}=L\mathcal{P}_{\max}\left[a\left(K+1\right)\log\left(K+1\right)+\left(1-a\right)K\log K\right],
\end{equation}
where $K=\left\lfloor (L\mathcal{P}_{\max})^{-1}\right\rfloor $,
$a=(L\mathcal{P}_{\max})^{-1}-K$ and $\left\lfloor \cdot\right\rfloor $
denotes the floor function, and
\begin{equation}
\label{bmax}
B_{\max}=S(Q)-\log L,
\end{equation}
with $S(Q)$ being the Shannon entropy of the probability vector 
\begin{equation}
Q=\frac{1}{LN}\{1+(LN-1)\sqrt{r},\underbrace{1-\sqrt{r},\ldots,1-\sqrt{r}}_{LN-1}\},
\end{equation}
given by $r=\left(LN\mathcal{P}_{\min}-1\right)/\left(LN-1\right)$.

\end{theorem}

The lower bound $B_{\min}$ is a direct extension of Theorem 2
established for mutually unbiased bases
by  Wu, Yu and M\o{}lmer in \cite{Molmer}.
 As our Theorem \ref{THP} generalizes and extends Theorem 1 from \cite{Molmer},
$B_{\min}$ is given as in Theorem 2 therein with their $C$ being
set to $L\mathcal{P}_{\max}$. Note that the results of Wu et al. 
were based on the detailed analysis performed by Harremo\"es and Topsoe  \cite{IEEE2001}. To derive
$B_{\max}$ we can directly rely on  \cite{IEEE2001}, using their Theorem II.5
part i). This result provides an upper bound for the sum of the Shannon
entropies as a function of the coefficient $\mathcal{P}$. Since this
bound is a decreasing function of $\mathcal{P}$, it remains valid
when $\mathcal{P}$ becomes substituted by its lower bound, namely
$\mathcal{P}_{\min}$. Note that the bound $B_{\max}$ is the  genuinely first result of that kind, while the alternative lower bounds can also be obtained by averaging the pairwise bounds (for $L=2$) or by multiobservable majorization \cite{RPZ14}. There is no possibility to get the pairwise counterpart of  $B_{\max}$ as for $L=2$ one has $\mathcal{M}_\textrm{min}=0$, $r=0$  and consequently $S(Q)=\log N$.

In the following sections we shall study several numerical examples showing the behavior of the optimal bounds in comparison with the analytical results at hand, including the progress described above.

\subsection{Three measurements for qubits}

A great circle is nondisplacable in a sphere, 
so any two such circles will always intersect. However, 
three great circles belonging to a sphere will generically not 
cross in a single point.
Therefore one can expect that for three
orthogonal measurements of a qubit  in three bases 
the average entropy of the probability vectors representing the measurement outcomes
will be less than the maximal value.

\begin{figure}
\centering
\includegraphics[width=1\linewidth]{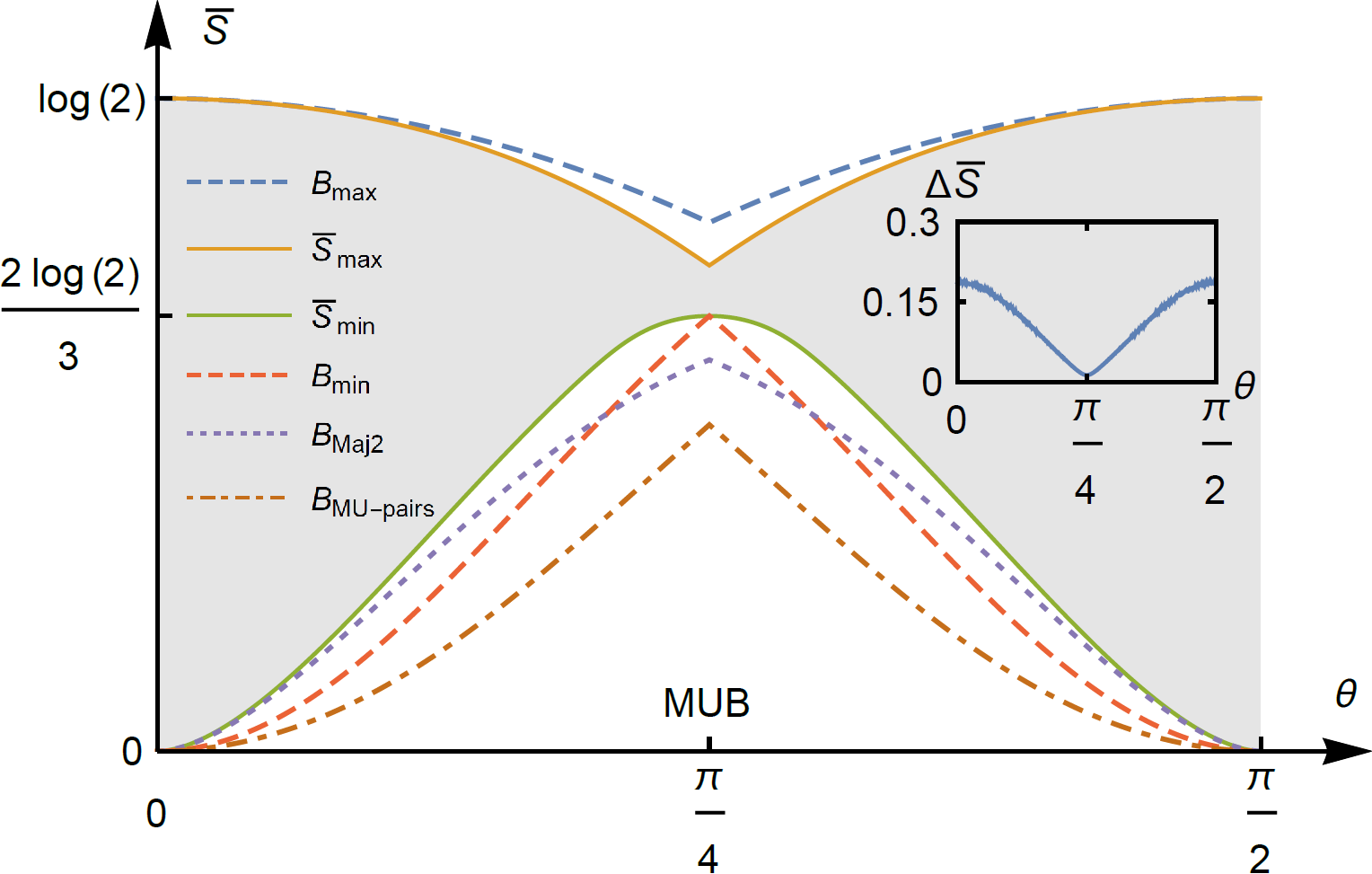}
\caption{As in Fig. \ref{fig:SN2L2}
 for $L=3$ measurements of a qubit.
Note the upper bound  $B_{\max}$ (\ref{bmax}) represented by the upper dashed curve,
which provides a nontrivial entropic certainty relation, and the lower bound $B_{\rm min}$ (27) which becomes tight at
$\theta=\pi/4$ for MUB.
}
\label{fig:SN2L3}
\end{figure}

To investigate this issue we analyze a one--parameter family of
three measurements, determined by three unitary matrices,
depending on an angle  $\theta$, 
\begin{equation} \label{eqn:3-mubs}
U_1={\1}_2, 
\ \
U_2=
\left(
\begin{smallmatrix}
\cos \theta & \sin \theta \\
\sin \theta & -\cos \theta
\end{smallmatrix}
\right)
,
\ \
U_3=
\left(
\begin{smallmatrix}
\cos \theta & \sin \theta \\
i \sin \theta & -i \cos \theta
\end{smallmatrix}
\right).
\end{equation}
Note that  for $\theta=0$ all three bases coincide, 
while in the case of $\theta = \pi /4$ they  become mutually unbiased.
Figure \ref{fig:SN2L3} shows the average entropy of measurement in these
three bases as a function of the angle $\theta$: 
the shaded area shows the allowed region bounded by solid lines, 
where dotted (or dashed-dotted) lines denote bounds obtained by using Maassen-Uffink 
relation \cite{MU88} and the majorization bound \cite{RPZ14}. Dashed lines correspond to the bounds
(\ref{bmin}) and (\ref{bmax}) provided by Theorem 2.
Note that the difference between 
the upper and the lower limits, computed numerically 
and represented by solid lines, is the smallest for
$\theta = \pi /4$, corresponding to MUBs. 
A similar property holds for the root mean square deviation
of the entropy, $\Delta {\bar S}$, presented in the inset.

In order to explore the generic case of arbitrary three orthogonal measurements of a qubit,
we work in first basis once more setting $U_{1}={\1}_2$
and draw remaining two matrices $U_{2}$ and $U_{3}$ according to the Haar measure
on the unitary  group $\mathcal{U}(2)$.
%
In Fig. \ref{fig22} we present the maximal and minimal values
of the average entropy $\bar S$ optimized over the set of all pure states
for collection of three randomly chosen bases. Variable $\xi$
at the horizontal axis characterizes the average deviation of the 
unitary transformation matrices from identity
and is normalized as $0\leq\xi\leq1$. It is defined by 
\begin{equation}
\xi^{2}=\frac{4}{3}\sum_{j=1}^{3}v_{j}(1-v_{j}),
\label{xixi}
\end{equation}
where the probabilities:
$v_{1}=|(U_2)_{11}|^2=\cos^{2}\theta_{1}$, 
   $v_{2}=|(U_3)_{11}|^2 = \cos^{2}\theta_{2}$ and
$
v_{3}= |(U_2U_3^{\dagger})_{11}|^2 =|\cos\theta_{1}\cos\theta_{2}+\sin\theta_{1}\sin\theta_{2}e^{i(\beta_{1}-\beta_{2})}|^{2},
$
are expressed as functions of phases entering the parameterization of unitary matrices
$U_2=[\cos\theta_1, \sin\theta_1; -\sin\theta_1, \cos\theta_1]$ 
and
$U_3=[e^{-i \beta_2}\cos\theta_2, -e^{-i \beta_1}\sin\theta_2; e^{i \beta_1}\sin\theta_2, e^{i \beta_2}\cos\theta_2]$. 

\begin{figure}
\begin{centering}
\includegraphics[scale=0.35]{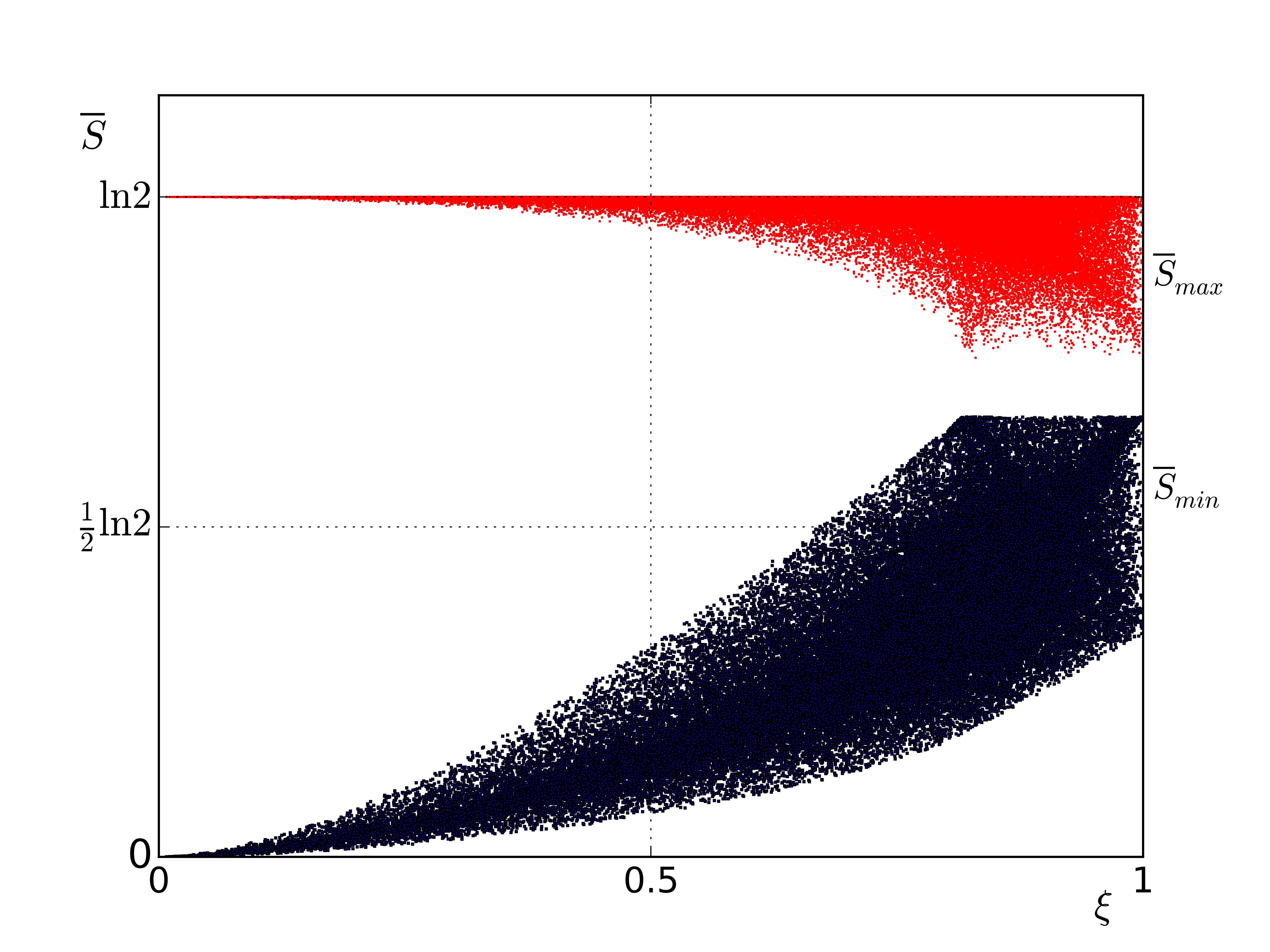}
\par\end{centering}
\protect\caption{Maximal and minimal value of the averaged
entropy $\bar S$ for a triple of random unitary gates of size $N=2$
as a function of the parameter $\xi$, 
which characterizes the average deviation of unitary matrices 
from identity, and equals to unity for MUB.}
\label{fig22}
\end{figure}

The value $\xi=0$ corresponds to the trivial case $\theta_{1}=0=\theta_{2}$,
for which $U_{2}=\1$ and $U_{3}$ is a phase gate. The opposite value
$\xi=1$ describes the case of MUB, for which $\theta_{1}=\pi/4=\theta_{2}$
and $\beta_{1}=\pi/4=\beta_{2}$. 
For $L=3$ MUBs of size $N=2$ the known bounds for the average entropy
\cite{Sa93} are tight:
 $B_{{\rm min}}^{SR}= \frac23 \log 2$ and 
\begin{equation}
B_{{\rm max}}^{SR} = \frac{1}{2} \log (6)- \frac{1}{2 \sqrt{3}} \log
   \left(2+\sqrt{3}\right) \approx 0.516, \label{upperMUB} 
\end{equation}
respectively. Note that $\log 2\approx 0.693$, giving in this case
the ultimate upper bound attained by mutually coherent states is 
larger than the optimal value (\ref{upperMUB}). In other words
the collection of three MUBs for $N=2$ does not share any mutually
coherent state. The bound $B_{\max}$, Eq. (\ref{bmax}), provides a reliable upper limitation, especially beyond the MUB case.

A closer look at  Fig. \ref{fig22} reveals  
that the trivial lower bound equal to $0$ is attained only if $\xi=0$,
as for other values of $\xi$ the non-trivial entropic uncertainty relations
apply. On the other hand, the maximal upper bound $\log2$, being the signature of
mutual coherence, is saturated for every allowed value of $\xi$. This
asymmetry is a major qualitative difference between mutual coherence
and $0$-entropy case, or in different terms, between certainty and
uncertainty relations. While the latter situation is typically forbidden
by quantum mechanics, the former case is rather common. When one comes
closer to mutual unbiasedness of the bases in question, it is however
more likely to find examples which possess no mutually coherent state
at all.

Observe that the difference between the upper and the lower limits
is the smallest for $\xi=1$, corresponding to MUB.
Numerical results show that a similar property holds for
the variance of the entropy.
We are not in position to prove this fact analytically for the Shannon
entropy. However, a related statement formulated in terms of the variance
of the Tsallis entropy of order two holds for any dimension $N$,
for which a complete set of $N+1$ MUBs exists - see Appendix~\ref{app_var}.

\subsubsection{Geometrical intuition on the Bloch sphere}
\begin{figure}
\begin{centering}
\includegraphics[scale=0.40]{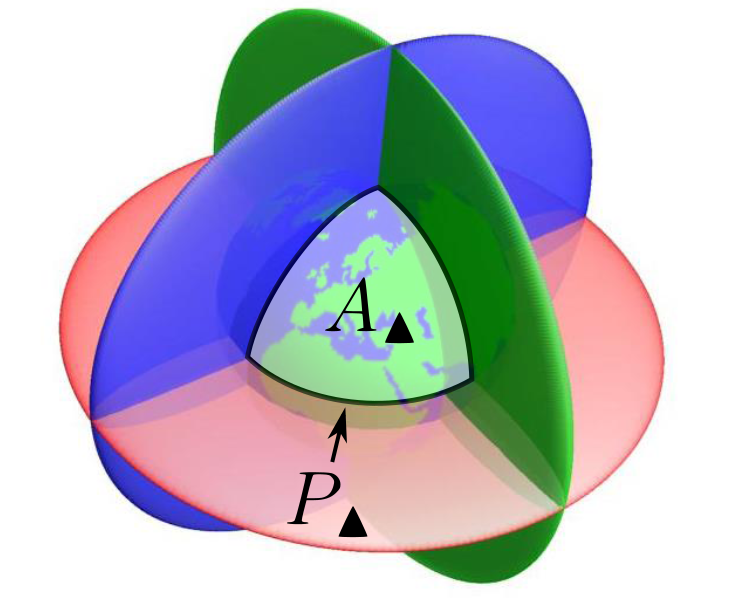}
\par\end{centering}
\protect\caption{Three equators on the sphere usually do not cross
    in a single point. The logo of the International Conference
   on Squeezed States and Uncertainty Relations (ICSSUR) held in Gda\'nsk in 2015,
  can be interpreted as a triple of MUBs for one qubit.
  Any such collection of orthogonal bases in ${\cal H}_2$
  can be associated with triangles on a sphere characterized by the minimal
 area $A_\blacktriangle$ or the minimal perimeter $P_\blacktriangle$.
  In the case of MUB shown here, all eight spherical triangles are of equal shape and size.}
\label{figICSSUR}
\end{figure}

\begin{figure}
\begin{centering}
\includegraphics[scale=0.22]{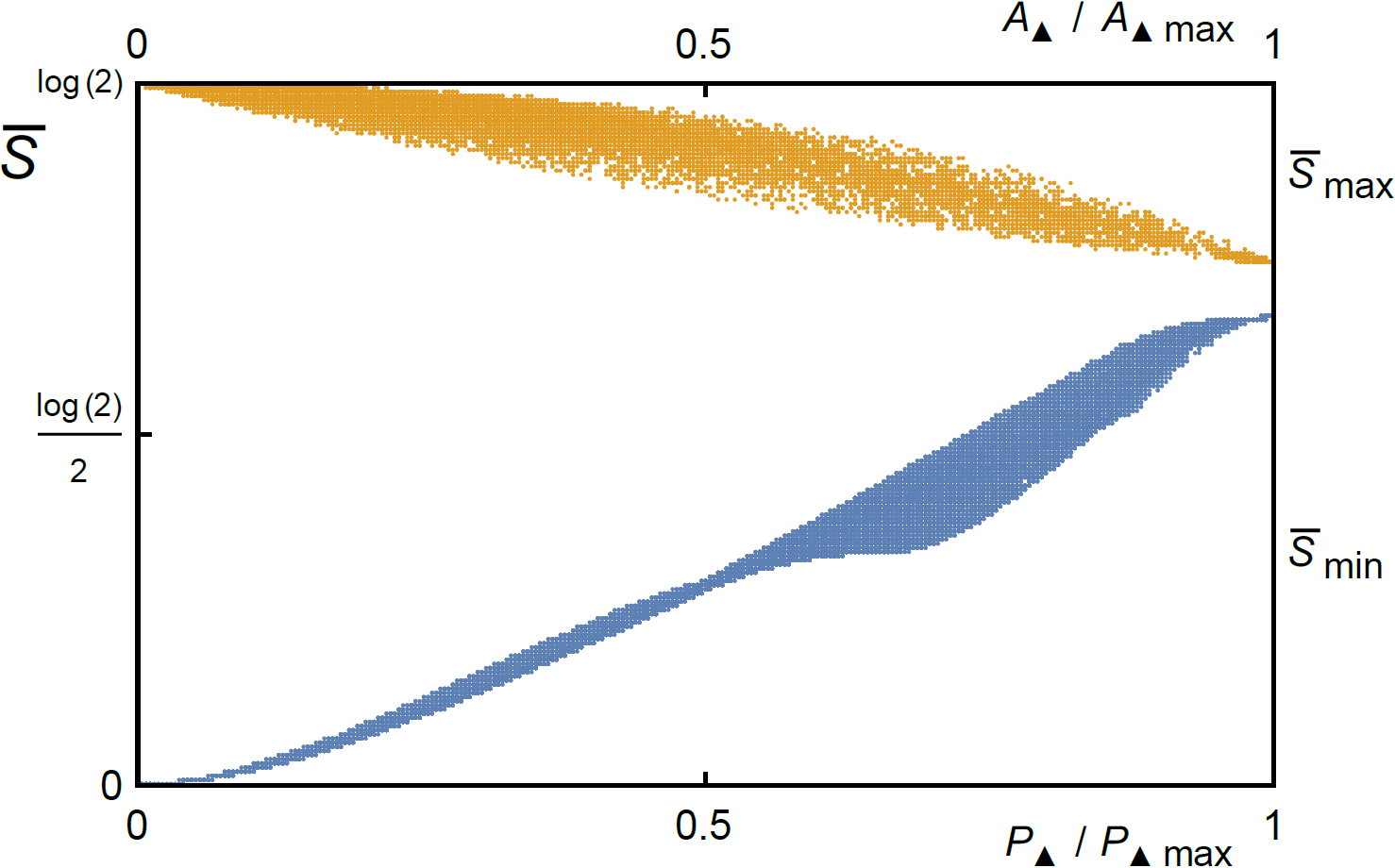}
\par\end{centering}
\protect\caption{Limits of the average entropy {$\bar S$} for $L=3$
   measurements   optimized over pure states from ${\cal H}_2$. The parameter $\bar S_\textrm{max}$ (upper abscissa; yellow) is depicted as a function of
   the smallest area  $A_\blacktriangle$ of the  spherical
triangle while $\bar S_\textrm{min}$ (lower abscissa; blue) is a function of the smallest perimeter  $P_\blacktriangle$. Both parameters are equal zero if  the three bases coincide
    and they attain their maxmal values $A_{\blacktriangle\, \max}$ and $P_{\blacktriangle\,
\max}$ for a set of MUBs.}
\label{figArea}
\end{figure}

Any orthogonal basis in ${\cal H}_2$ can be represented as a pair of antipodal points on the Bloch sphere. With three bases ($3$ pairs of antipodal points)
one can thus associate (generically) eight spherical triangles laying on the Bloch sphere.
 Let $A_\blacktriangle$ and $P_\blacktriangle$ denote respectively the smallest area and the smallest perimeter calculated among all these triangles. Both parameters are equal to zero
  if all three bases do coincide, and achieve the maximum, if the three
  bases in question are mutually unbiased. Such a MUB case with
  $A_{\blacktriangle\, \max}=\pi/2$ and $P_{\blacktriangle\,\max}=3 \pi/2$ is sketched in Fig. \ref{figICSSUR}.
 Moreover, we shall  observe that  the geometric parameters $A_\blacktriangle$ and $P_\blacktriangle$
   are invariant  with respect to any unitary rotation of the reference
frame.

  For any triple of random unitary matrices of order $N=2$ we found (repeating the calculations leading to Fig. \ref{fig22})
  the parameters $A_\blacktriangle$ and $P_\blacktriangle$, and further computed extremal
  values of the mean entropy $\bar S$ optimized over the set of pure states.
  Results presented  in Fig. \ref{figArea} show that the area of the minimal triangle
  carries information concerning the upper bound for the mean entropy
   while the smallest perimeter characterizes the lower bound.
We observe that the proposed geometrical invariants (area and perimeter) reliably capture the property of mutual unbiasedness visible as the narrow entropy window on the right hand side of the plot.

\subsection{$N+1$ measurements in $N$ dimensions}

Consider a family of four bases in ${\cal H}_3$, 
determined by the following unitary matrices 
\begin{equation}
\label{eqn:4-mubs}
U_1={\1}_3, 
\ \
U_2= (\mathcal{F}_3)^{4 \theta/\pi} ,
\end{equation}
\begin{equation}\label{eqn:4-mubs-cd}
U_3= D (\mathcal{F}_3)^{4 \theta/\pi},
\ \
U_4= D^2 (\mathcal{F}_3)^{4 \theta/ \pi}.
\end{equation}
Here $\mathcal{F}_3$ represents the Fourier matrix of size three,
while $D = {\rm  diag} (1,\exp(i 2 \pi /3),\exp(i 2 \pi /3))$.
As in the one--qubit case, all matrices become diagonal for $\theta =0$ and correspond to the same basis,
while for $\theta=\pi/4$ the bases are mutually unbiased.

\begin{figure}
\centering
\includegraphics[width=1\linewidth]{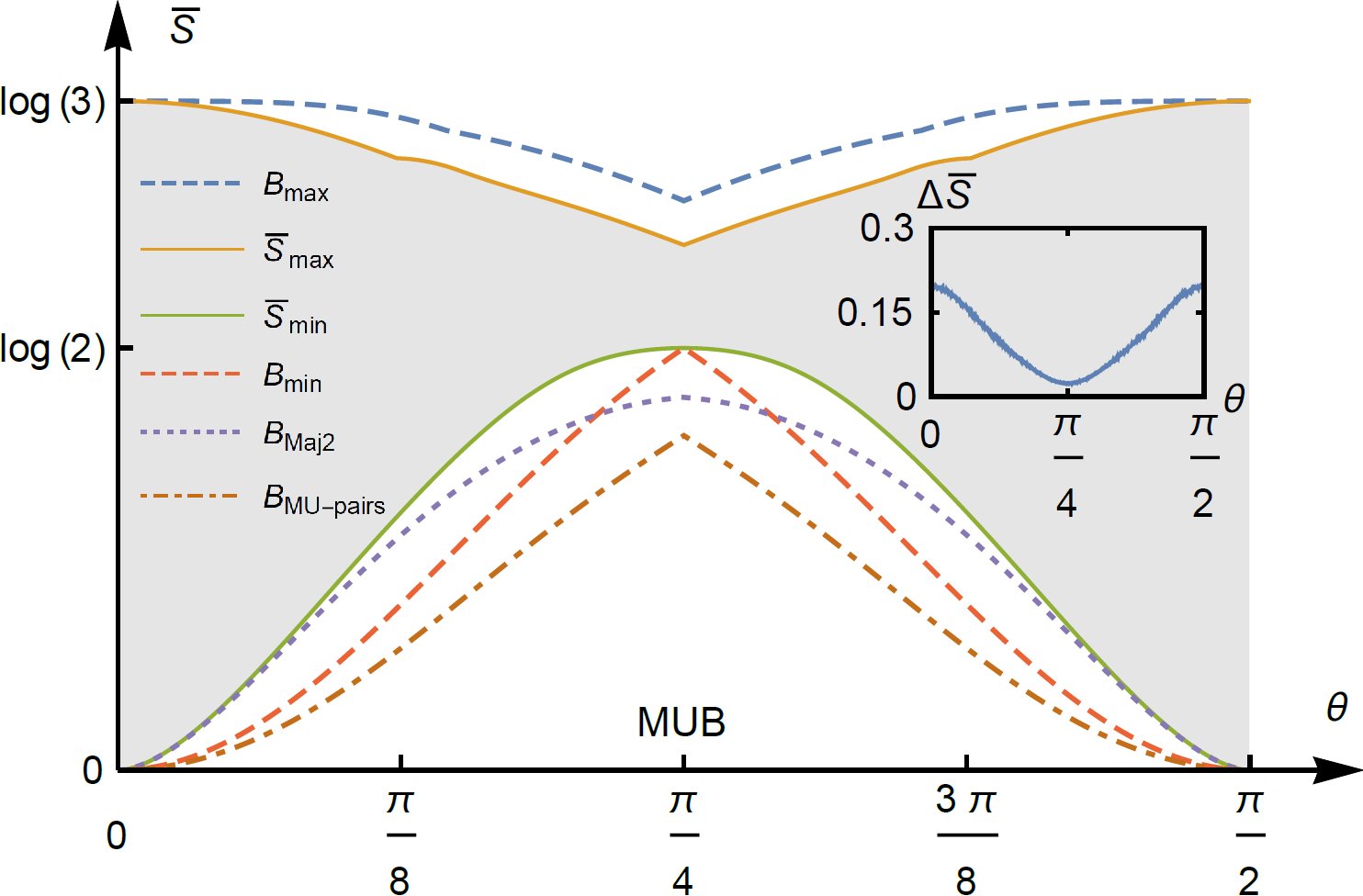}
\caption{As in Fig. \ref{fig:SN2L3}
 for $L=4$ orthogonal measurements in $N=3$ dimensions.
Note a nontrivial behavior of the maximal value,
${\bar S}_{\rm max}$ and the upper bound $B_{\rm max}$, 
which attain their minima for $\theta=\pi/4$ corresponding to MUB.
 }
\label{fig:SN3L4}
\end{figure}

Fig.~\ref{fig:SN3L4} presents the behavior of numerically  computed
maximal and minimal values of the average entropy $\bar S$ 
compared with analytical bounds.
Note that the difference between 
the numerical upper and the lower limits, which are represented by solid lines, is once more the smallest for
$\theta = \pi /4$, corresponding to MUBs. 
A similar property holds as well for the root mean square deviation
of the entropy, $\Delta {\bar S}$, presented in the inset.

Let us now proceed to larger dimensions of the Hilbert space.
In this place we are going to restrict our attention to 
prime power dimensions, $N=p^k$, for which a set 
of $N+1$ MUBs is known \cite{Woott87,DEBZ10}.
In this very case concrete upper and lower bounds for the
average entropy $\bar S$ were obtained by Sanchez--Ruiz \cite{Sa93,Sa95},
\begin{equation}
\begin{split}
&B_{\min}^{SR} \! =
\left\{
\begin{array}{cc}
\log \frac{N+1}{2} &\text{ for $N$ odd}\\
\frac{N}{2 (N+1)} \log \frac{N}{2} + \frac{N/2 + 1}{N+1} \log\left(\frac{N}{2} +1 \right) &\text{ for $N$ even}
\end{array}
\right.  ,
\\
\\
&B_{\max}^{SR} = \log N + \frac{(N-1)^2 \log (N-1)}{(N+1)N (N-2)}.
\end{split}
\label{SR2}
\end{equation}
and later generalized by Wu, Yu and M\o{}lmer in \cite{Molmer}.
Observe that both bounds asymptotically behave as
$\log N - {\rm const}$, where the constant reads 
$a_{\min}^{SR}= \log 2  \approx  0.693$ for the lower bound 
and it vanishes for the upper bound, i.e.  
$a_{\max}^{SR}=0$.

\begin{figure}
\centering
\includegraphics[width=1\linewidth]{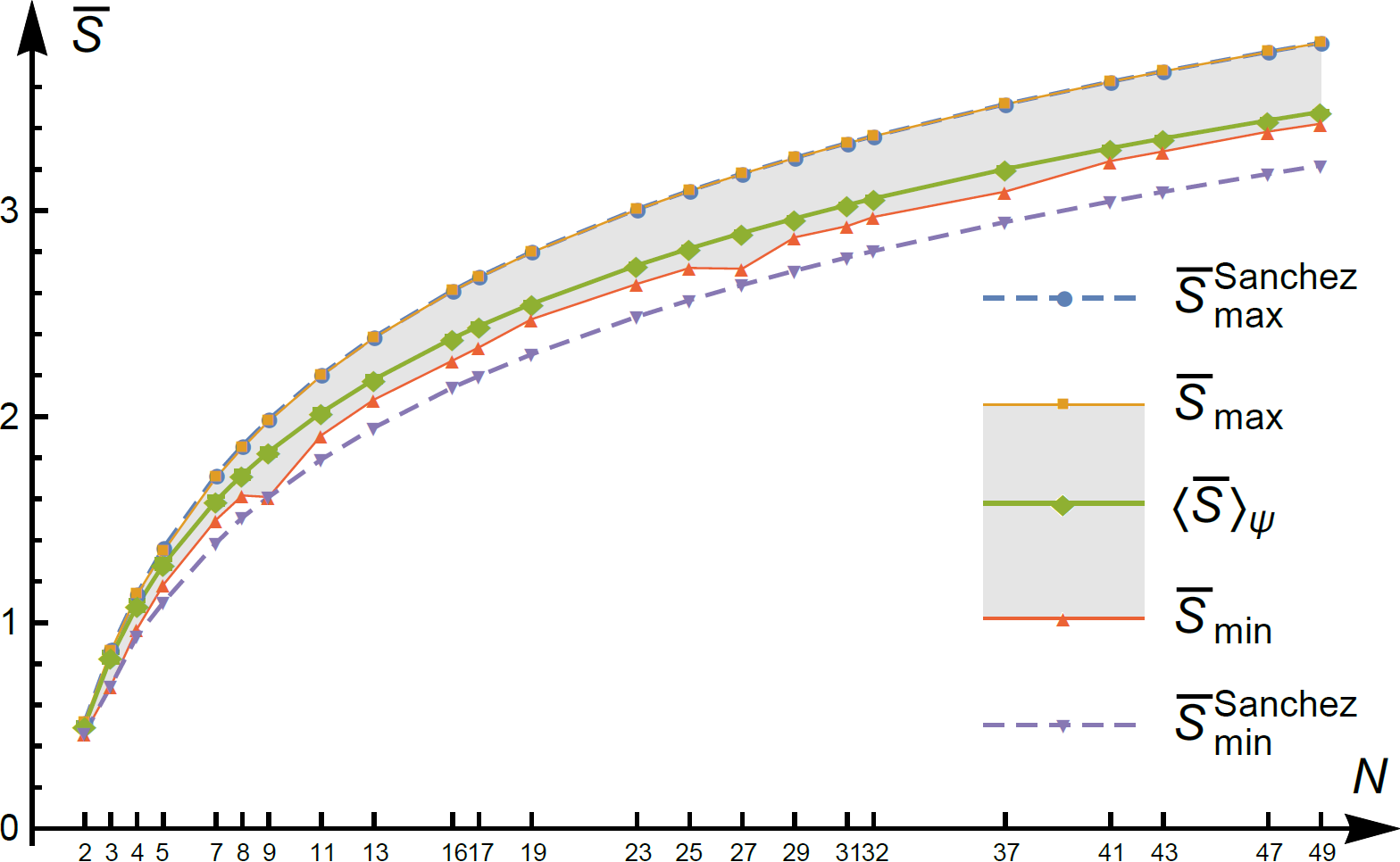}
\caption{Behavior of the average entropy $\overline{S}$ over a set MUBs in $\mathcal{H}_N$ as a function of the dimensions $N$
for power of primes.
Upper and lower bounds of Sanchez-Ruiz are compared with numerical maximum and minimum taken 
over the set of all pure states.}
\label{fig:MUB}
\end{figure}

Figure \ref{fig:MUB} shows both bounds (dashed lines) compared with numerically obtained
lower and upper limits, ${\bar S}_{\rm min}$, and ${\bar S}_{\rm max}$. The central curve
shows the behavior of the mean value $\langle {\bar S} \rangle_{\psi}$, averaged
over entire set of pure states in ${\cal H}_N$ with respect to the unitarily invariant Haar measure.
For dimensions  $N$ of the order of $20$  the error bars, marked in the graph, 
are smaller than the symbol size. Note that the allowed, shaded region,
is very close to the upper bound of Sanchez. This suggests that the bound 
$B_{\max}^{SR}$ is close to optimal, while it is more likely
to improve the lower bound $B_{\min}^{SR}$.

The mean Ingarden-Urbanik entropy of a random pure state is given by~\cite{jones1990entropy}
\begin{equation}
\langle S^{IU}\rangle = \Psi(N + 1) - \Psi(2) \underset{N \to \infty}{\simeq} \log N - (1- \gamma),
\end{equation}
where $\Psi$ is the digamma function and $\gamma\approx 0.577$ is the Euler Gamma constant. 
Unitary invariance of a random state $\ket{\psi}$ gives us, that for a complete
set of MUBs in dimension $N = p^k$ we have
\begin{equation}
\begin{split}
\left \langle \frac{1}{N+1} \sum_{i=1}^{N+1} S^{IU}(\ket{\psi},U_i)\right\rangle 
& \ \ = \Psi(N + 1) - \Psi(2)
\\ &\underset{N \to \infty}{\simeq}  \log N - (1- \gamma).
\end{split}
\end{equation} 

Even though we were in position to study the problem for dimensions
not exceeding $50$, we found it interesting to analyze limiting 
behavior of our results. All three numerical curves can be fitted with 
a general relation $S_j \approx  \log N -{\rm const}$, 
where the fitted value are
$a_{\min} \approx  0.48$, $a_{\max} \approx 0.07$,
and $a_{\rm av} \approx  0.42$ for the average over all pure states.
The latter value coincides well with the asymptotic result 
$a_{\rm \infty} = 1-\gamma \approx  0.422$,
while the  former values contribute to the conjecture that the lower analytical bound (\ref{SR2})
of Sanchez might be easier to improve.

\section{mutually entangled states\label{sec:mutually-entangled-states}}

In the second part of this work we link uncertainty and certainty relations for the average measurement entropy 
with quantum entanglement related to different splittings of the 
composite Hilbert space.

Entanglement of any pure state of a bipartite system $|\psi\rangle\in{\cal H}_{A}\otimes{\cal H}_{B}$
can be characterized by its \textsl{entropy of entanglement} equal
to the von Neumann entropy of the partial trace $E(|\psi\rangle):=S(\rho_{A})$,
where $\rho_{A}={\rm Tr}_{B}|\psi\rangle\langle\psi|$. In analogy
with the notions presented in the previous sections we shall discuss
the \textsl{mutual entanglement} of a given state with respect to
various splittings of the $N\times N$ composite system (we assume
here that both ${\cal H}_{A}$ and ${\cal H}_{B}$ have the dimension
$N$). A direct counterpart of the uncertainty relations (\ref{sumx}) is
\begin{equation}
0\le E_{{\rm min}}\le {\bar E}\le E_{{\rm max}}\leq \ln N,\label{sumy}
\end{equation}
where the mutual entanglement, averaged with respect to $L$
different splittings of the Hilbert space, reads
\begin{equation}
{\bar E}:=\frac1L \sum_{j=1}^{L}E\bigl(W_{j}|\psi\rangle\bigr),
\label{mutent}
\end{equation}
and $\left\{ W_{j}\right\} $ with $j=1,\dots,L$ 
denotes a collection of $L$ bipartite
unitary gates,   i.e. unitary matrices of order $N^{2}$.
 Note that similarly to the case of ordinary
uncertainty relations it is convenient to set $W_{1}=\1$.

\begin{figure}
\begin{centering}
\includegraphics[scale=0.33]{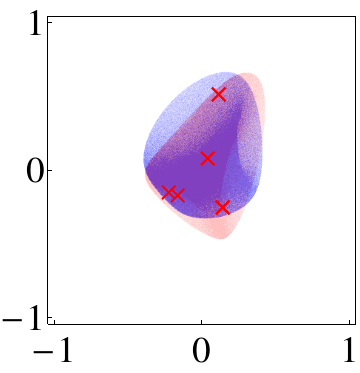}\,\,\,\includegraphics[scale=0.33]{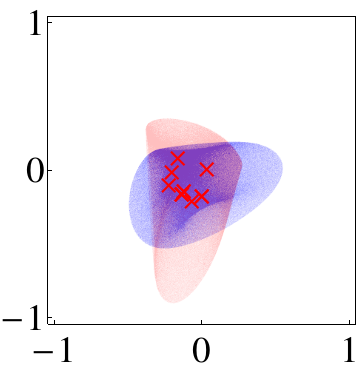}
\par\end{centering}
\protect\caption{
Two examples of projections of $\Real P^3$ and $U (\Real P^3)$ embedded in $\Cplx P^3$ on a plane.
Projection of intersection points marked by crosses correspond to the mutually entangled states.
}
\label{figrp3}
\end{figure}

We are now going to study the simplest case of $N=2$ and $L=2$,
which deals with two splittings only. Apart from the original splitting
given by the computational product basis $|i,j\rangle$ ($i,j=1,2$),
there is the second splitting described by the transformed basis $W_{2}|i,j\rangle$.
We have the following:

\begin{prop} Consider the case $N=2$ and $L=2$, and an arbitrary
unitary matrix $W_{2}\in \mathcal{U}(4)$. Then a) the upper bound in (\ref{sumy})
is saturated as there exists a \textsl{mutually entangled} state $|\psi_{{\rm ent}}\rangle$,
so that $E_{{\rm max}}= \log 2$; b) the lower bound in (\ref{sumy})
is saturated as well and there exists a \textsl{mutually separable}
state $|\psi_{{\rm sep}}\rangle$, so that $E_{{\rm min}}=0$.
\end{prop}
A proof of this proposition based on canonical form of a two-qubit
gate \cite{KC01,HVC02} is provided in Appendix A.
To show part b) of the above proposition one can also rely on
geometric properties of projective spaces. Let us recall
that in the two-qubit case the manifold of maximally entangled states
is $\mathcal{U}(2)/\mathcal{U}(1)=\mathbb{R}P^{3}$. As it is known that $\mathbb{R}P^{M}$
is nondisplacable in $\mathbb{C}P^{M}$  with respect to transformations
by $\mathcal{U}(M+1)$ \cite{Oh93,Tam08},
two real manifolds embedded into a complex one have to 
intersect -- see Fig. \ref{figrp3}.
In a particular case $M=3$, 
the presence of the intersection points directly implies 
that for any choice of $W_2\in \mathcal{U}(4)$
there exists a \textsl{mutually entangled state} $|\psi_{{\rm ent}}\rangle$
maximally entangled with respect to the partition of the Hilbert space in 
the computational basis $W_{1}=\1$ and in the basis rotated by $W_2$.

For any family of $L=2$ unitary matrices $W_1, W_2$ of order $4$
the upper bound for the averaged entanglement is saturated, 
${\bar E}=\log 2$.  Therefore we show in Fig. \ref{fig:MEB,n=2}
the averaged entanglement ${\bar E}$ for a family of $L=3$
unitary matrices of order $4$ 
explicitly given later in Eq. (\ref{eqn:MEB2-family}). This plot, analogous to Fig. \ref{fig:SN2L3},
displays nontrivial upper bounds for ${\bar E}$.
Moreover, these results suggest that there exists a state mutually 
separable with respect to all three splittings of ${\cal H}_4$
into ${\cal H}_2 \otimes {\cal H}_2$.

Numerical results obtained for the dimension $N=3$ suggest that \emph{for any}
unitary gate $W_{2}$ of order $9$, there exists a related mutually
entangled state, so that we conjecture that in general $E_{{\rm max}}=\log N$. To prove
this conjecture it would be enough to show that the space of maximally
entangled states $\mathcal{U}(N)/\mathcal{U}(1)$ is non-displacable in $\mathbb{C}P^{N^{2}-1}$
with respect to action of $\mathcal{U}(N^2)$.
Note that the dimension of this space is $N^{2}-1$, and equals the
\textsl{half} of the dimension of the embedding space, as it forms
a Lagrangian manifold. Since the similar scenario occurred for quantum coherences,  it is thus tempting to conjecture that the above statement,
true if $N=2$, holds also for any $N\ge3$.

\section{Mutually entangling gates\label{sec:Mutually-entangling-gates}}

In the previous section we introduced the concept of mutual entanglement
and studied this notion in the simplest case of two different splittings
of the composite Hilbert space. Now we aim to consider an arbitrary
number of $L\ge3$ bi-partite unitary matrices $W_{j}$, $j=1,\dots,L$
(with $W_{1}=\1$), which define different tensor-product structures.
Since one copes with $L$ different splittings of the entire system
into subsystems, one can define the notion of separable and maximally
entangled states with respect to these partitions and ask about
the quantum states for which $\bar E$ given in Eq. (\ref{mutent}) is
minimal or maximal. 

The same approach can be applied
for instance in the particular case $N=4$, as the system consists
of two ququarts or rather four qubits $A,B,C,D$. 
For instance,  setting $L=3$ and choosing $W_{2}$
and $W_{3}$ to be suitable permutation matrices, 
which define bipartite splittings
$AB|CD$, $AC|BD$ and $AD|BC$, respectively, one can study the
mutual entanglement with respect to different partitions and look
for maximally entangled multipartite states \cite{Sc04,FFPP08,AC13,GZ14} such that all their reductions are maximally mixed.
In the case of four qubits, there are no pure states, maximally entangled with
respect to three above partitions \cite{GBP98,HS00}.

In the case of bi--partite unitary gates one distinguishes \textsl{special
perfect entanglers}, which transform a product basis into maximally
entangled basis \cite{Re04}. More formally, a unitary matrix $W$ acting
on ${\cal H}_{N}\otimes{\cal H}_{N}$ will be briefly called an \textsl{entangling
gate}\textsl{\emph{,}} if all its columns are maximally entangled \cite{Co11},
so it transforms separable basis states
into maximally entangled states 
\begin{equation}
E(W|i,j\rangle)\ =\ \log N,\ \ {\rm for}\ i,j=1,\dots,N.
\label{entagate}
\end{equation}
Such gates are known for any $N$ \cite{We01,WGC03}, so for $L=2$
there exists a gate for which the minimal mutual entanglement ${\bar E}_{\textrm{min}}$
will not be smaller than $\frac{1}{2}\log N$. Quite interestingly, for two--qubit
systems such gates are especially distinguished, as they maximize
the \textsl{entangling power}, i.e. the average entropy of entanglement
produced from a generic separable state \cite{ZZF00}.

Analyzing the case of a larger number of unitaries $L\ge3$, we are
going to demonstrate the existence of \textsl{mutually entangling
gates}, able to transform product states into states maximally entangled
with respect to all $L$ splittings in question. In other words, these
unitary matrices are formed out of maximally entangled vectors, which
remain maximally entangled in any transformed splitting. In a direct
analogy to the notion of mutually unbiased bases \cite{DEBZ10} we
define \textsl{mutually entangled gates}.

\begin{definition} \label{def:MEB}
We say that a collection of unitary matrices $W_1=\id, W_2, \dots, W_L
\in\mathcal{U}(N^{2})$ is  \textsl{mutually entangled} if for $i \neq j$ the gates
$W_i^{\dagger}W_j$ satisfy condition (\ref{entagate}), i.e. the columns of the 
matrix $W_j$ are maximally entangled in the basis given by $W_i$ and vice versa. 
\end{definition}
We shall also say, that the
columns of these unitary matrices form \emph{mutually entangled bases}
(MEB). As is shown below, both concepts happen to be closely related.

\begin{theorem} If there exists a set of $m$ MUBs in ${\cal H}_{N}$,
then there also exists a set of $m$ MEBs in ${\cal H}_{N}\otimes{\cal H}_{N}$.
\end{theorem}
In other words, Theorem 3, states that $m$ mutually unbiased bases
provide the set of $m$ mutually entangling gates, for which 
the average entanglement ${\bar E}$ satisfies
\begin{equation}
\frac{m-1}{m} \log N \leq {\bar E} \leq \log N.
\end{equation}

We prove the above theorem by constructing the relevant entangling gates. 
First we recall the construction of unitary bases  by Werner~\cite{We01}, 
called 'shift and multiply'. For a given Latin square
$\{\lambda(j,k)\}_{j,k=1}^N$ and a collection of Hadamard matrices
$H^{(1)},H^{(2)}, \dots , H^{(N)}$ one constructs unitary matrices
\begin{equation}
U^{(i,j)} = \sum_{k=1}^N H^{(j)}_{i,k} \ket{\lambda(j,k)} \bra{k} ,
\end{equation}
which form an orthogonal basis of the 
Hilbert-Schmidt space of complex matrices of order $N$.
Thus the columns of a matrix 
\begin{equation}
V = \frac{1}{\sqrt{N}}\sum_{k,i,j=1}^{N} H^{(j)}_{i,k} \ket{\lambda(j,k), k} \bra{i , j}
\end{equation}
form a maximally entangled basis in  $\mathbb{C}^{N^2}$.
Note, that $V$ can be written as ($T$ denotes the transposition)
\begin{equation}
V = P \left(H^{(1)}{}^T \oplus H^{(2)}{}^T \oplus \dots \oplus H^{(N)}{}^T\right) U_{\rm SWAP},
\end{equation}
where $U_{\rm SWAP}$ is a swap permutation matrix
and $P$ is a permutation matrix given by
\begin{equation} \label{eqn:latin-permutation}
P = \sum_{k,l} \ketbra{\lambda(l,k),k}{l,k}.
\end{equation}

The rows of the matrix $V$ do not generate a maximally entangled basis,
but if we permute its columns and define 
\begin{equation}
W = P \left(H^{(1)}{}^T \oplus H^{(2)}{}^T \oplus \dots \oplus H^{(N)}{}^T\right) P^T,
\end{equation}
then the rows and columns of the matrix $W$ generate a maximally entangled basis.
The above reasoning leads to the following 
explicit construction of mutually entangled gates.

Assume that we are given a Latin square $\{\lambda(j,k)\}_{j,k=1}^N$
and  a collection of $k$ mutually unbiased bases $M_1,M_2,\dots,M_k$ of size $N$.
The bases are unbiased, that is
\begin{equation}
M_i^{\dagger} M_j \text{\ is a rescaled Hadamard matrix for } i \neq j.
\end{equation}
Using these matrices we introduce a collection of bases
\begin{equation}
W^{(i)} =  P (\id \otimes M_i) P^T.
\end{equation}
We have the following:
\begin{corollary}
Let $\lambda$ be a Latin square of size $N$ and let $P$ be defined as
in~\eqref{eqn:latin-permutation}. Then the bases $W^{(i)}$ are mutually entangled.
\end{corollary}
\begin{proof}

We write for $i \neq j$.
\begin{equation}
\begin{split}
W^{(i)}{}^{\dagger} W^{(j)} &= (P (\id \otimes M_i) P^T)^{\dagger} P (\id \otimes M_j) P^T\\
&=  P (\id \otimes M_i^{\dagger} M_j) P^T.
\end{split}
\end{equation}
Since $ M_i^{\dagger} M_j$ is a rescaled Hadamard matrix we obtain, that $W^{(i)}{}^{\dagger}
W^{(j)}$ is a unitary basis.
\end{proof}
To demonstrate how the above construction works in action
we provide in the Appendix B the two collections of mutually
entangled bases, respectively  for $2\times 2$ and $3 \times 3$ systems.

\subsection{Mutual entanglement for two--qubit system} 
Let us consider a family of matrices defined in Eq.~\eqref{eqn:3-mubs}
which interpolates between $\{\id,\id,\id\}$ for $\alpha =0$ and MUB for $\alpha = \pi/4.$
From the above family we construct bases of $\mathbb{C}^4$ as 
\begin{equation} \label{eqn:MEB2-family}
\begin{split}
\Big\{W_0&=\1_4,
W_1= \left(
\begin{smallmatrix}
 \cos (\alpha ) & 0 & 0 & \sin (\alpha ) \\
 0 & \cos (\alpha ) & \sin (\alpha ) & 0 \\
 0 & \sin (\alpha ) & -\cos (\alpha ) & 0 \\
 \sin (\alpha ) & 0 & 0 & -\cos (\alpha ) \\
\end{smallmatrix}
\right), \\
W_2 &=\left(
\begin{smallmatrix}
 \cos (\alpha ) & 0 & 0 & \sin (\alpha ) \\
 0 & \cos (\alpha ) & \sin (\alpha ) & 0 \\
 0 & i \sin (\alpha ) & -i \cos (\alpha ) & 0 \\
 i \sin (\alpha ) & 0 & 0 & -i \cos (\alpha ) \\
\end{smallmatrix}
\right)
\Big\}.
\end{split}
\end{equation}
In the case of $\alpha = \pi/4$ the above family forms mutually entangled bases.
We analyzed lower and upper bounds for the
average entanglement ${\bar E}$ with respect to three splittings of ${\cal H}_4$,
as  defined in (\ref{mutent}).
%
In the case of $L=2$ splittings 
the upper bound, ${\bar E}_{\rm max}=\log 2$,
can be saturated, 
but for $L=3$ the upper bound becomes not trivial -- see Fig.~\ref{fig:MEB,n=2}.
%
%
The average entanglement ${\bar E}$ attains its minimal value, 
${\bar E}^{MEB}_{\rm max}$  given by \eqref{upperMUB},
for three mutually entangled bases corresponding to $\alpha = \pi/4.$

\begin{figure}
\centering
\includegraphics[width=0.87\linewidth]{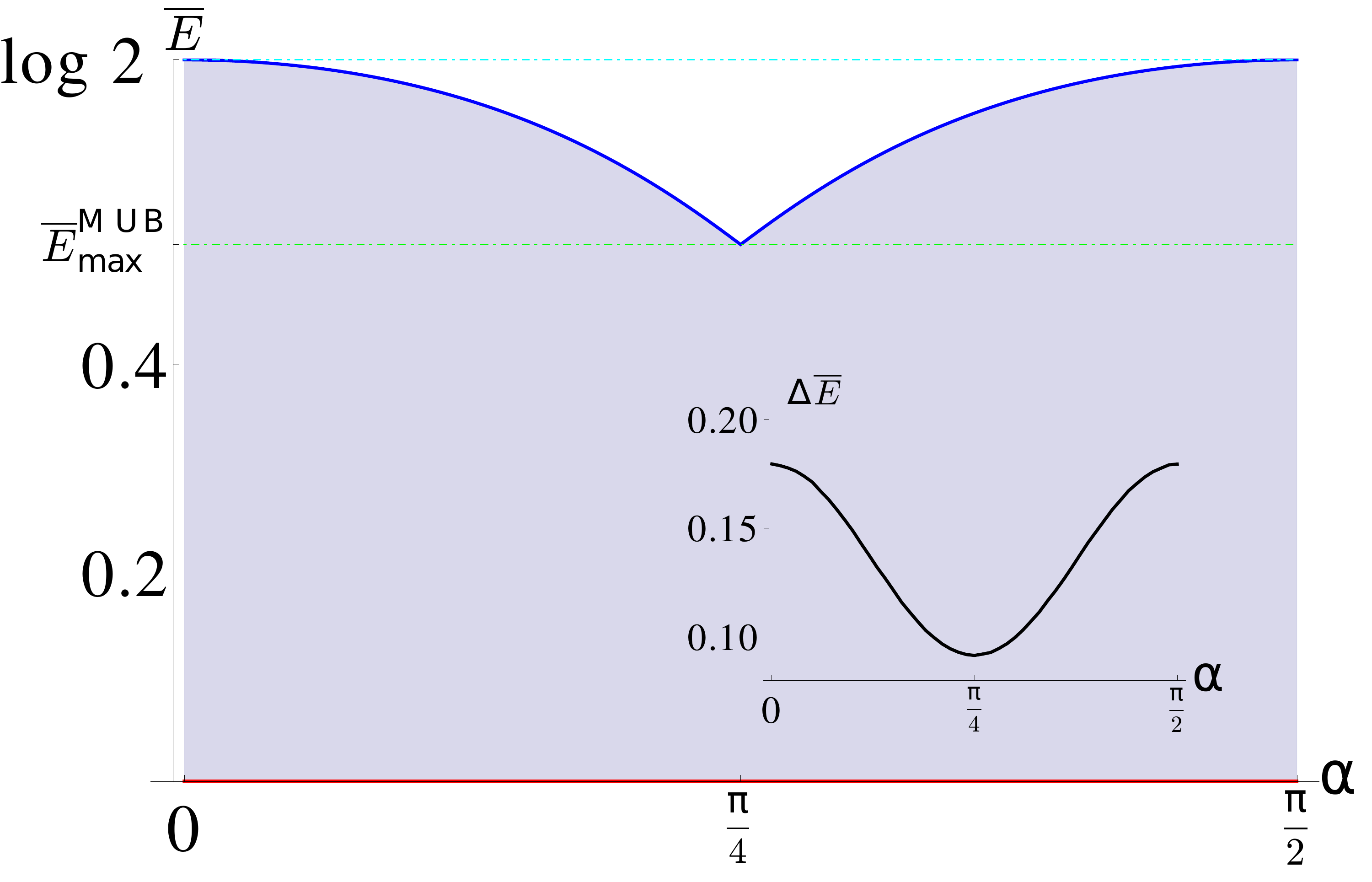}
\caption{
Average entanglement $\bar E$ for family of three unitary gates~\eqref{eqn:MEB2-family} of size $4$
as a function of the phase $\alpha$. 
For $\alpha = \pi/4$ corresponding to MEB the upper bound for $\bar E$ attains its minimum, 
as well as the root square deviation $\Delta {\bar E}$ shown in the inset. 
}
\label{fig:MEB,n=2}
\end{figure}


\section{Concluding remarks}

A lot of work was recently done to improve and generalize
entropic uncertainty relations, which provide lower bounds
for the average entropy of probability vectors describing  measurements
in several orthogonal bases. Following the ideas of Sanchez 
\cite{Sa93} in this work we analyzed in parallel also upper bounds for the
average entropy, and obtained entropic uncertainty (\ref{bmin}) and certainty (\ref{bmax}) relations
valid for an arbitrary number of measurements of any pure state in ${\cal H}_N$.
Main motivation for such a study stems from a search of states which are 
simultaneously unbiased with respect to bases determining orthogonal measurements.
Such states display the effects of quantum coherence with respect to all these bases,
so that the average entropy becomes maximal.

In the case of any  $L=2$ measurements in an arbitrary $N$--dimensional 
Hilbert space mutually coherent states exist, so the upper bound for the 
average entropy is saturated, ${\bar S}=\log N$. 
This result, related to non-displacablity of the great torus
in complex projective space $\mathbb{C}P^{N-1}$~\cite{Cho04},
does not hold for a larger number of $L\ge 3$ measurements,
for which certainty relations become non-trivial. Numerical results 
show that the analytic upper bound derived for MUBs
by Sanchez  \cite{Sa93} is rather precise, so 
it would be desirable to generalize them for other collections
of orthogonal bases.

Analyzing probabilities obtained in sequence
of $L$ orthogonal measurements of a quantum state 
one can also interpreted them
as a result of a single generalized measurement $P$, called positive
operator valued measure (POVM), which consist of $N\cdot L$ projection
operators. Hence the averaged entropy (\ref{sumx}) of $L$ orthogonal
measurements is equal, up to an additive constant $\log L$, to the
entropy of the probability vector describing the POVM. Furthermore,
the so-called {\sl informational power} of $P$ \cite{DAS11} 
associated to the set of MUBs, is closely related with the minimal
entropy ${\bar S}_{\rm min}$, averaged over $L=N+1$ measurements and minimized
over the set of all pure states. This quantity occurs to be equal
to $\log N-{\bar S}_{\rm min}$ \cite{SS04,DA14,Sz04}, 
while the quantity $\log N-{\bar S}_{\rm max}$
coincides with the minimal relative entropy. 

A complete set of $L=N+1$ mutually unbiased bases in ${\cal H}_N$
forms an optimal scheme of a quantum measurement 
distinguished be several statistical properties \cite{DEBZ10}.
Our numerical results allow us to conjecture 
that the complete set of MUBs minimizes fluctuations 
of the average entropy while varying the pure state investigated.  
\begin{conjecture} For any choice of $L=N+1$ measurements in a dimension
$N=p^{k}$ the lower bound for the averaged entropy $\bar S$ 
achieves its maximum and the upper bound achieves its minimum if $L$
unitary matrices form a MUB.
\end{conjecture}
\begin{conjecture} The standard deviation of the averaged entropy 
$\Delta \bar S=\sqrt{\langle {\bar S}^{2}\rangle_{\psi}-\langle  {\bar S} \rangle_{\psi}^{2}}$,
averaged over the entire set of pure states of size $N$ is minimal
if the collection of $L=N+1$ unitary matrices forms a MUB.
\end{conjecture}
Not being able to prove conjecture 2 for the Shannon entropy
we provide in Appendix C a proof of an analogous proposition 
formulated in terms of the Tsallis entropy of order two.
This result contributes to our understanding of the special 
role mutually unbiased bases play in the theory of quantum measurement.

The second key goal of this work was to establish a closer
link between entropic uncertainty relations and the theory of
quantum entanglement. For any composed Hilbert space
${\cal H} = {\cal H}_{N} \otimes {\cal H}_{N}$,
the corresponding product basis $|i,j\rangle$ 
and a global unitary gate $U\in \mathcal{U}(N^2)$
one can investigate entanglement with respect to 
the transformed bases, $U|i,j\rangle$, with $ i,j=1,\dots N$.
For any pure state $|\psi\rangle$ of a bipartite
system we analyzed its average entanglement $\bar E$
with respect to several choices of the 
separable bases, linked by unitaries $U_1, \dots, U_L$,
and investigated lower and upper bounds for this quantity.

In the case of two--qubit system the average entanglement
for $L=2$ can attain the limiting value $\log 2$, 
as a state {\sl mutually entangled} 
with respect to both splittings exists.
This result follows from the fact that the set of 
two--qubit maximally entangled states, 
equivalent to the real projective space $\mathbb{R}P^{3}$, 
is non-displacable in $\mathbb{C}P^{3}$
with respect to the action of $\mathcal{U}(4)$.
Numerical results allow us to conjecture that a similar
statement holds also in higher dimensions. 

It is worth to emphasize that nondisplacability of real 
projective spaces in the corresponding complex projective space \cite{Tam08} 
admits  other applications. Consider 
$N=3$ dimensional space corresponding to 
angular momentum $j=(N-1)/2=1$
and the set $\cal C$ of $\mathcal{SU}(2)$--coherent states
obtained by the rotating the 
maximal weight state $|j,j\rangle=|1,1\rangle$
by Wigner rotation matrix \cite{ZFG90}.
In the stellar representation these states
are described by two stars coinciding into a 
single point of the sphere. The set $\cal A$ of 'anticoherent states',
which are as far from $\cal C$ as possible,
contains the state $|1,0\rangle$ represented 
by two stars in antipodal points at the sphere.
Hence the set $\cal A$ has the form of the real projective space
$\mathbb{R}P^{2}$, which is non-displacable in $\mathbb{C}P^{2}$
with respect to the action of $\mathcal{U}(3)$.
This implies that the sets $\cal A$ and ${\cal A}'=U({\cal A})$
do intersect, so there exists a pure state 
anticoherent with respect to any two choices 
of the maximal weight state.

Let us conclude the paper with a short list of 
open questions. It is a challenge to improve 
explicit 'certainty relations': {\sl upper}
bounds for the average entropy obtained
for $L\ge 3$ measurements
with respect to arbitrary orthogonal bases. 
In the case of MUB the upper bounds of Sanchez \cite{Sa93}
occur to be rather precise, so it is more likely 
to improve his lower bounds. It would be 
interesting to derive analogous lower and upper bounds 
for the averaged entanglement of a bipartite state
with respect to $L\ge 3$ different splittings 
of the Hilbert space and to prove existence of 
mutually entangled states for the general  
$N \times N$ problem.

\medskip

Acknowledgments. It is our pleasure to thank  I. Bengtsson,
P. Horodecki, M. Ku\'{s}, P. Nurowski, M. Plenio,
W. S{\l}omczy{\'n}ski, A. Szymusiak  and L. Vaidman 
for fruitful discussions.  We are grateful to Magdalena Stobi{\'n}ska and Marek {\.Z}ukowski
for giving us an opportunity to present these results during the International
Conference on
Squeezed States and Uncertainty Relations ICSSUR15 held in Gda\'nsk. We also thank the organizing comitee of ICSSUR2015
for letting us include and modify the inspiring logo of this conference (see Fig. \ref{figICSSUR}). We acknowledge financial support by the Polish National Science
Centre (NCN) under the grant number DEC-2012/04/S/ST6/00400 (Z.P)
and the John Templeton Foundation (K.{\.Z}).
\L .R. acknowledges financial support by the grant number 2014/13/D/ST2/01886
of the National Science Center, Poland. Research in Freiburg is supported
by the Excellence Initiative of the German Federal and State Governments
(Grant ZUK 43), the Research Innovation Fund of the University of
Freiburg, the ARO under contracts W911NF-14-1-0098 and W911NF-14-1-0133
(Quantum Characterization, Verification, and Validation), and the
DFG (GR 4334/1-1).

\appendix


\section{Mutually entangled states and mutually separable states for two qubits}
In this Appendix we demonstrate existence of 
mutually entangled states and mutually separable states
in the two qubit case.
Without loss of generality we may consider two matrices $\id$ and
$W_2 \in \mathcal{U}(4)$, which is brought by local unitary transformations into 
its canonical form \cite{KC01,HVC02},
\begin{equation}\label{eqn:unitary-canonical}
\! W_2 \! = \! \! \left(
\begin{smallmatrix}
 e^{i b_3} \cos \left(b1\right) & 0 & 0 & i e^{i b_3} \sin \left(b_1\right) \\
 0 & e^{-i b_3} \cos \left(b_2\right) &  i e^{- i b_3} \sin \left(b_2\right)  & 0 \\
 0 &   i e^{- i b_3} \sin \left(b_2\right) & e^{-i b_3} \cos \left(b_2\right) & 0 \\
 i e^{i b_3} \sin \left(b_1\right) & 0 & 0 & e^{i b_3} \cos \left(b_1\right)
\end{smallmatrix}
\right)\!\!\!\!
\end{equation}
parameterized by three real parameters $b_1,b_2,b_3$. 
Next we find vectors $\ket{x}$ and $\ket{y}$ such that 
\begin{equation}
\begin{split}
\frac12 E(\ket{x}) + \frac12 E(W_2 \ket{x}) &= 0 \\
\frac12 E(\ket{y}) + \frac12 E(W_2 \ket{y}) &= \log 2,
\end{split}
\end{equation} 
i.e. $\ket{x}$ is separable in both bases and $\ket{y}$ is maximally entangled
in both bases. 

We see immediately, that we can take $\ket{y} =  (\ket{0,0} + \ket{1,1})/\sqrt{2}$.
To show  existence of a mutually separable vector we consider two cases.
If $b_2 = 0$ we may take $\ket{x} = \ket{0,1}$ and in opposite case we may take 
$\ket{x}$ to be proportional to 
\begin{equation}
\ket{x} \simeq  \ket{0} \otimes   
\left(
\ket{0} + 
e^{2 i b_3} \sqrt{\frac{ \sin(b_1) \cos(b_1) }{ \sin(b_2) \cos(b_2)}} \ \ket{1}
\right).
\end{equation}

\section{Examples of mutually entangled bases}

We provide here exemplary collections of three
unitary matrices of order $2^2$ and four unitary matrices of order $3^2$, which
form mutually entangled bases (see Definition~\ref{def:MEB}).

A collection of three mutually entangled bases for two qubits reads,
\begin{equation}
W_1=\id_4 ,
W_2=
\frac{1}{\sqrt{2}} \left(
\begin{smallmatrix}
 1 & 0 & 0 & 1 \\
 0 & 1 & 1 & 0 \\
 0 & 1 & -1 & 0 \\
 1 & 0 & 0 & -1 \\
\end{smallmatrix}
\right),
W_3=
\frac{1}{\sqrt{2}}
\left(
\begin{smallmatrix}
 1 & 0 & 0 & 1 \\
 0 & 1 & 1  & 0 \\
 0 & i & -i & 0 \\
 i & 0 & 0 & -i
\end{smallmatrix}
\right).
\label{meb2}
\end{equation}

In the case of two-qutrit system, four mutually entangled bases are:
\begin{equation}
\begin{split}
&W_1 =  \id_9,\quad W_2=
 \frac{1}{\sqrt{3}}\left(
\begin{smallmatrix}
 1 & 0 & 0 & 0 & 1 & 0 & 0 & 0 & 1 \\
 0 & \omega  & 0 & 0 & 0 & \omega^2  & 1 & 0 & 0 \\
 0 & 0 & \omega & 1 & 0 & 0 & 0 & \omega^2 & 0 \\
 0 & 0 & 1 & 1 & 0 & 0 & 0 & 1 & 0 \\
 1 & 0 & 0 & 0 & \omega & 0 & 0 & 0 & \omega^2 \\
 0 & \omega^2 & 0 & 0 & 0 & \omega & 1 & 0 & 0 \\
 0 & 1 & 0 & 0 & 0 & 1 & 1 & 0 & 0 \\
 0 & 0 & \omega^2 & 1 & 0 & 0 & 0 & \omega & 0 \\
 1 & 0 & 0 & 0 & \omega^2 & 0 & 0 & 0 & \omega 
\end{smallmatrix}
\right)\!,\, W_{3,4}=\\
&
\frac{1}{\sqrt{3}}\left(
\begin{smallmatrix}
 1 & 0 & 0 & 0 & 1 & 0 & 0 & 0 & 1 \\
 0 & \omega^2 & 0 & 0 & 0 & 1 & \omega & 0 & 0 \\
 0 & 0 & \omega^2 & \omega & 0 & 0 & 0 & 1 & 0 \\
 0 & 0 & 1 & 1 & 0 & 0 & 0 & 1 & 0 \\
 \omega & 0 & 0 & 0 & \omega^2 & 0 & 0 & 0 & 1 \\
 0 & 1 & 0 & 0 & 0 & \omega^2 & \omega & 0 & 0 \\
 0 & 1 & 0 & 0 & 0 & 1 & 1 & 0 & 0 \\
 0 & 0 & 1 & \omega & 0 & 0 & 0 & \omega^2 & 0 \\
 \omega & 0 & 0 & 0 & 1 & 0 & 0 & 0 & \omega^2
\end{smallmatrix}
\right), 
\frac{1}{\sqrt{3}}\left(
\begin{smallmatrix}
 1 & 0 & 0 & 0 & 1 & 0 & 0 & 0 & 1 \\
 0 & 1 & 0 & 0 & 0 & \omega & \omega^2 & 0 & 0 \\
 0 & 0 & 1 & \omega^2 & 0 & 0 & 0 & \omega & 0 \\
 0 & 0 & 1 & 1 & 0 & 0 & 0 & 1 & 0 \\
 \omega^2 & 0 & 0 & 0 & 1 & 0 & 0 & 0 & \omega \\
 0 & \omega & 0 & 0 & 0 & 1 & \omega^2 & 0 & 0 \\
 0 & 1 & 0 & 0 & 0 & 1 & 1 & 0 & 0 \\
 0 & 0 & \omega & \omega^2 & 0 & 0 & 0 & 1 & 0 \\
 \omega^2 & 0 & 0 & 0 & \omega & 0 & 0 & 0 & 1 
\end{smallmatrix}
\right)
\end{split}
\label{meb3}
\end{equation}
where $\omega = e^{ 2 \pi i /3 }$ is the cubic root of unity. 

\section{ Variance of measurement outcomes}
\label{app_var}

Consider $L$ orthogonal measurements performed 
on a quantum state of size $N$. Numerical results suggest 
that the variance of the average entropy characterizing the measurements 
is minimal if all measurement bases are mutually unbiased.
In this appendix we prove this statement for the Tsallis entropy or order two,
earlier used for purpose of uncertainty relations \cite{Rastegin,Ra15}.

The Tsallis entropy $T_\beta$  of order $\beta > 0$ is defined as
\begin{equation}
T_\beta (p) = \frac{1}{\beta - 1} \left( 1 - \sum p_i^\beta \right).
\end{equation}
and reduces to the Shannon entropy as $\alpha \to 1$.

Consider an arbitrary unitary matrix $U_i \in {\cal U}(N)$  defining a bases, 
in which an orthogonal measurement is performed.
For any pure state $|\psi\rangle$ we introduce 
the corresponding vector of probabilities
\begin{equation}
p^{(i)}_j = |\bra{j} U_i \ket{\psi}|^2,
\label{prob3}
\end{equation}
described by the Tsallis entropy $T_2$, 
\begin{equation}
T^{(i)} = T_2 (p^{(i)}) =1 - \sum_j \left(p_j^{(i)}\right)^2.
\label{tsal3}
\end{equation}
Assume now that $\ket{\psi}$ is 
a random pure state distributed according to the unitary invariant Haar
measure. We can now average the mean Tsallis entropy over the
entire set of pure states and analyze its variance.

\begin{theorem}
Let $U_1=\id, U_2, \dots , U_L$ be a collection of
$L$ unitary matrices of order $N$,  which for any 
state $|\psi\rangle$ leads to the set of probability vectors 
(\ref{prob3}) described by the Tsallis entropy (\ref{tsal3})
and its mean value
${\bar T}=(T^{(1)} + T^{(2)} + \dots + T^{(L)})/L$. 
If the set of $L$ MUBs in dimension $N$ exists
then the variance of the mean entropy, 
${\rm var} ( {\bar T})$ 
averaged over the set of all pure states in ${\cal H}_N$ 
is minimal if matrices $\{ U_j\}_{j=1}^L $ are mutually unbiased.
\end{theorem}

\begin{proof}
Note, that to prove this we can restrict our attention to the case of two
matrices unitary. For convenience we denote $U_1\equiv \id$, $U_2 \equiv U$ and 
\begin{equation}
p_i = |\scalar{i}{\psi}|^2, \qquad q_j = |\bra{j} U \ket{\psi}|^2.
\end{equation}
Next we write
\begin{equation}
\begin{split}
&{\rm var}( T_2 (p) + T_2 (q)) = 
\langle (T_2 (p) + T_2 (q))^2 \rangle- 
\langle  T_2 (p) + T_2 (q)\rangle^2 \\
&=
\langle  T_2^2 (p) \rangle  + \langle T_2^2 (q)\rangle 
 - \langle T_2 (p) + T_2 (q)\rangle^2 + 2 \langle T_2 (p) T_2 (q) \rangle . 
\end{split}
\end{equation}
Unitary invariance of the distribution of  $\ket{\psi}$
implies that the first three terms do not depend on $U$, so
to get the minimum value of the variance one has to minimize the last term 
$\langle  T_2 (p) T_2 (q) \rangle $. Let us rewrite it in the form 
\begin{equation}
\begin{split}
&\langle  T_2 (p) T_2 (q) \rangle  =\
\langle  (1 - \sum p_i^2) (1 - \sum q_i^2) \rangle \\
&= 1 -  \langle \sum p_i^2 \rangle  - \langle \sum q_i^2 \rangle + \langle \sum p_i^2 \sum q_j^2 \rangle .
\end{split}
\end{equation}
To get the minimum one should minimize the average $\langle \sum p_i^2 \sum q_j^2 \rangle $,
which consist of the following terms 
\begin{equation}
\begin{split}
\langle  p_i^2 q_j^2 \rangle 
=
\langle |\psi_i|^4 |(U\ket{\psi})_j|^4 \rangle .
\end{split}
\end{equation}
Treating the vector $\ket{\psi}$ as a first column of a random unitary matrix
distributed according to the Haar measure, 
we can use Weingarten calculus \cite{Weing} and obtain the following value
\begin{eqnarray}
&& \!\!\! \langle |\psi_i|^4 |(U\ket{\psi})_j|^4 \rangle = 
\frac{(N-1)! 4!}{(N+3)!}  
\Big(
|u_{ji}|^4 
+  |u_{ji}|^2 \sum_{k \neq i} |u_{jk}|^2 \nonumber
\\
& &
+  \frac16 \sum_{k=1}^N |u_{jk}|^4 
+  \frac16 \sum_{k\neq l \atop k,l\neq i}^N  |u_{jk}|^2  |u_{jl}|^2
\Big).
\end{eqnarray}
The above result implies that
\begin{widetext}
\begin{equation}
\begin{split}
\langle  \sum_i p_i^2 \sum_j q_j^2 \rangle 
&=
\frac{(N-1)! 4!}{(N+3)!}  
\left(
\left( 1 + \frac{(N-1)}{6} \right) \sum_{ij} |u_{ij}|^4
+
\left( 1 + \frac{(N-2)}{6} \right)
\sum_i \sum_{k\neq l} |u_{ik}|^2 |u_{il}|^2
\right)\\
&=
\frac{(N-1)! 4!}{(N+3)!}  
\left(
\frac16 \sum_{ij} |u_{ij}|^4 +
\left(1+ \frac{N-2}{6}\right) \sum_i \sum_{kl} |u_{ik}|^2 |u_{il}|^2
\right)\\
&=
\frac{(N-1)! 4!}{(N+3)!}  
\left(
\frac16 \sum_{ij} |u_{ij}|^4 +
\left(1+ \frac{N-2}{6}\right) N
\right).
\end{split}
\end{equation}
\end{widetext}
It is  now easy to conclude that the above expression is minimized for
$|u_{ij}|^2 =
1/N$, i.e. for $U$ being unbiased with identity. The same reasoning
applied for $L ( L - 1 ) / 2$ pairs of measurements 
gives us, that the variance of the sum of all measurements will be minimal if
all matrices are mutually unbiased.
\end{proof}

\end{document}